\documentclass[10pt,journal,compsoc,twocolumn]{IEEEtran}

\usepackage{stmaryrd}
\usepackage{amsfonts}
\usepackage{mathrsfs,amsmath}
\usepackage{amssymb}
\usepackage{makeidx}
\usepackage{multirow}
\usepackage{algorithmic} 
\usepackage{bm}
\usepackage{setspace}
\usepackage{amsthm}
\usepackage{empheq}
\usepackage{booktabs}
\usepackage[ruled,linesnumbered,vlined]{algorithm2e}  

    \SetCommentSty{mycommfont}
\usepackage{color}
\usepackage{enumerate}
\usepackage[left=0.65in, right=0.65in, top=0.7in, bottom = 0.98in]{geometry}
\usepackage{makecell}
\usepackage{threeparttable}
\usepackage{booktabs}
\usepackage{tabularx}
\usepackage{subcaption} 
\usepackage{array}
\usepackage[font=footnotesize]{caption}
\usepackage{tabularx,booktabs}
\usepackage{ifpdf}
 \ifpdf
 \else
 \fi
 
\usepackage{cite}
 
\ifCLASSINFOpdf
   \usepackage[pdftex]{graphicx}
\else
\fi
 
\usepackage{amsmath}
\usepackage[percent]{overpic}

\usepackage{array} 
\usepackage{tikz,pgfplots,filecontents}
\usetikzlibrary{spy} %
\usetikzlibrary{patterns}
\pgfplotsset{width=8.6cm, height=5cm, compat=1.9}
 
\usepackage{url}
\usepackage{hyperref}

\makeatletter
\let\NAT@parse\undefined
\makeatother

\newtheorem{theorem}{Theorem}
\newtheorem{lemma}[theorem]{Lemma}

\newtheorem{remark}[theorem]{Remark}

\IEEEoverridecommandlockouts   

\allowdisplaybreaks

\pgfplotsset{every axis legend/.style={%
    cells={anchor=west},
    inner xsep=3pt,inner ysep=2pt,nodes={inner sep=0.8pt,text depth=0.15em},
    anchor=north east,%
    shape=rectangle,%
    fill=white,%
    draw=black,
    at={(0.98, 0.98)},
    font=\footnotesize,
    }
}

\pgfplotsset{every axis/.append style={line width=0.6pt,tick style={line width=0.8pt}}}

\title{Online Learning for Intelligent Thermal Management of Interference-coupled and Passively Cooled Base Stations} 

\begin{document}

\author{Zhanwei~Yu,~\IEEEmembership{Student Member,~IEEE,}
        Yi~Zhao~\IEEEmembership{Student Member,~IEEE,}~Xiaoli~Chu,~\IEEEmembership{Senior~Member,~IEEE}, and~Di~Yuan,~\IEEEmembership{Senior~Member,~IEEE}
\IEEEcompsocitemizethanks{\IEEEcompsocthanksitem Z. Yu, Y. Zhao, and D. Yuan (\{zhanwei.yu; yi.zhao; di.yuan\}@it.uu.se) are with the Department of Information Technology, Uppsala University, Uppsala, Sweden.
\IEEEcompsocthanksitem X. Chu (x.chu@sheffield.ac.uk) is with the Department of Electronic and Electrical Engineering, University of Sheffield, Sheffield, UK.
}}


\IEEEtitleabstractindextext{%
    \begin{abstract} Passively cooled base stations (PCBSs) have
    emerged to deliver better cost and energy efficiency. However,
    passive cooling necessitates intelligent thermal control via
    traffic management, i.e., the instantaneous data traffic or
    throughput of a PCBS directly impacts its thermal
    performance. This is particularly challenging for outdoor
    deployment of PCBSs because the heat dissipation efficiency is
    uncertain and fluctuates over time.  What is more, the PCBSs are
    interference-coupled in multi-cell scenarios. Thus, a
    higher-throughput PCBS leads to higher interference to the other
    PCBSs, which, in turn, would require more resource consumption to
    meet their respective throughput targets. In this paper, we
    address online decision-making for maximizing the total downlink
    throughput for a multi-PCBS system subject to constraints related on
    operating temperature.  We demonstrate that a reinforcement
    learning (RL) approach, specifically soft actor-critic (SAC), can
    successfully perform throughput maximization while keeping the
    PCBSs cool, by adapting the throughput to time-varying heat
    dissipation conditions. Furthermore, we design a denial and reward
    mechanism that effectively mitigates the risk of overheating
    during the exploration phase of RL. Simulation results show that
    our approach achieves up to 88.6\% of the global optimum. This is
    very promising, as our approach operates without prior knowledge
    of future heat dissipation efficiency, which is required by the
    global optimum.  \end{abstract}

    \begin{IEEEkeywords}
        Reinforcement learning, interference, passive cooling, throughput maximization, thermal management. 
    \end{IEEEkeywords}
}

\maketitle

\IEEEdisplaynontitleabstractindextext

%
\IEEEpeerreviewmaketitle

\section{Introduction}

The rapid expansion of 5G networks has increased data, but also energy consumption. As a result, operators face higher operational expenses. The study in \cite{alliance2021network} shows that active cooling consumes up to 40\% of the total energy usage by mobile networks. Recently, the concept of passively cooled base stations (PCBSs) has emerged~\cite{ericsson22}. Unlike traditional base stations, a PCBS uses passive cooling solutions, such as highly thermally conductive fillers (e.g., graphene, carbon nanotubes), to keep the baseband unit (BBU) within an acceptable temperature range, thus completely eliminating the energy cost due to active cooling. Moreover, PCBSs are also easier to install, reducing the deployment cost.

By its nature, managing the thermal performance becomes crucial for PCBS. The heat generation of a PCBS' BBU is directly mapped to the amount of traffic, or throughput, that the PCBS is delivering. Ideally, we would like to have a PCBS deliver as much throughput as possible while avoiding overheating. This is not an easy task for two reasons. First, the heat dissipation efficiency is an uncertain system parameter for an outdoor PCBS, as it is influenced by ambient temperature, wind speed, humidity, and solar exposure, etc. Second, for a multi-cell scenario, the PCBSs are interference-coupled. Delivering more throughput at one PCBSs means also more frequent transmissions, thus increasing the interference to the other PCBSs, which would have to allocate more time-frequency resource, to meet their respective throughput targets.  Hence, we not only need to consider the thermal matter, but also have to ensure that the resource limit is adhered to, accounting for the fact that the consumption of transmission resource of a PCBS depends also on the throughput set by the other PCBSs.

In this paper, we consider an outdoor multi-PCBS system deploying
orthogonal frequency-division multiple access (OFDMA) and
multiple-input multiple-output (MIMO) technology. The goal is to
maximize the sum throughput of PCBSs in the downlink, subject to the
BBU temperature limit and interference coupling of the PCBSs. This is
an online decision-making problem.  The BBU temperature evolution is
determined by the PCBS' throughput (i.e., the amount of traffic
served) and the heat dissipation efficiency.  The latter is an
uncertain system parameter that varies over time.  Hence we need to
deal with throughput maximization without knowing future heat
dissipation efficiency.  Our contributions to dealing with this online
optimization problem are as follows:

\begin{enumerate}
 \item We propose an online PCBS thermal management method based on a reinforcement learning (RL) approach, specifically soft actor-critic (SAC)~\cite{haarnoja2018soft}. The method maximizes the total downlink throughput adaptively and dynamically, subject to the communication resource and BBU temperature constraints.

\item Because of interference coupling, the consumption 
of time-frequency resource of a PCBS, referred to as load, is affected
by, in addition to the PCBS's own throughput, the load of other
PCBSs. As a result, throughput maximization cannot be performed
separately for each PCBS.  We model and analyze this load-coupling
system for MIMO transmission, facilitating the evaluation of an action
with respect to the required transmission resource in SAC.

\item Since exploration in the RL framework may cause overheating of the BBU, we provide a refined denial and reward mechanism to ensure that the RL approach can sufficiently explore while preventing BBU overheating. This, in fact, enables the RL approach to be used online without any pre-training.

\item Simulation shows very promising results. Specifically, the proposed approach delivers up to 88.6\% of the global optimum.  Note that achieving the global optimum is not possible in practice, as it requires solving a problem offline with complete information of heat dissipation efficiency in future (i.e., an oracle), whereas our method makes decisions online without prior knowledge of future heat dissipation efficiency. In addition, the learned policy has good generalization with respect to user location update.
\end{enumerate}


\section{Literature Review}

The works in \cite{liu2023passive, sui2023membrane, Aslan2019Passive, Aslan2019Heat, Duan2021Thermal} have studied passive cooling techniques for base stations. In \cite{liu2023passive}, the authors provide a passive cooling method using sorption-based evaporative cooling with a salt-embedded composite sorbent for 5G communication equipment, to significantly enhance heat dissipation performance. Similarly, in \cite{sui2023membrane}, the authors propose a passive thermal management strategy that relies on moisture desorption from hygroscopic salt solutions through a protective membrane that only allows water vapor to pass through. The research in \cite{Aslan2019Passive} conducted thermal simulations to evaluate the cooling performance of distinct CPU heatsinks. In \cite{Aslan2019Heat}, a novel heat dissipation structure has been designed with antennas and digital beamforming chips on opposite sides, demonstrating superior cooling efficiency compared to traditional designs. Paper \cite{Duan2021Thermal} explores various factors for efficient PCBS cooling, including environmental influences, strategic placement, and the use of innovative materials. Furthermore, some works focus on the optimization for passively cooled planar arrays. In \cite{Aslan1}, a system-driven convex algorithm is proposed for the synthesis of passively cooled planar phased arrays of mm-wave 5G base stations. In \cite{Aslan2}, a convex optimization algorithm is derived to synthesize the array layouts with minimized side lobes within a pre-defined cell sector.

There are works on thermal management in data centers and wireless communication systems through resource allocation. The study in \cite{data-center-1} addresses a holistic energy minimization problem in a data center with an active cooling system.  The server temperature is influenced by the thermal load, server inlet cold air temperature, and airflow rate. To ensure reliability, the server temperature must be kept below a threshold, which can be achieved by adjusting parameters such as fan speed. In addition, in a BBU pool of a cloud radio access network (C-RAN), the authors of \cite{c-ran-1} provide an algorithm for dynamic allocation of computing resource, integrating thermal characteristics in the optimization of power consumption and computing resource usage, such that the processors work at the proper temperature. In~\cite{c-ran-2}, the research highlights the challenge of minimizing computational power while optimizing computing resource usage efficiency under thermal constraints, due to that the static power grows with the operating temperature of core computing unit. The authors present an algorithm that dynamically adjusts voltage and frequency to save power while meeting traffic request deadlines. In \cite{Ran2023}, the authors propose an optimization framework based on deep reinforcement learning, to jointly optimize energy consumption from the perspectives of task scheduling and cooling control in data centers. Note that the computing resources (i.e., processors or BBU) in data centers and C-RAN are centralized in a place far from users. In contrast, thermal management in base stations needs to consider additional matters, such as user demand and signal processing, making it more complex. For example, the work in~\cite{Bao2023Thermal} studies task schedule and resource allocation in a UAV-and-Base-station hybrid-enabled mobile network with a passively cooled system, and with joint optimization of user admission policy, task scheduling, UAV hovering trajectory, traveling time, and computation-and-communication resource allocations.

In addition, both \cite{alsuhli2021mobility} and \cite{9855432} study traffic and load management of base stations using RL. The study in \cite{alsuhli2021mobility} introduces a deep RL framework for mobility traffic management in cellular networks, utilizing cell individual offset adjustments. Depending on network size, it employs enhanced deep Q-learning or advanced RL methods like deep deterministic policy gradient and twin delayed deep deterministic policy gradient. In \cite{9855432}, the authors introduce a multi-objective RL framework for traffic balancing for multi-band downlink cellular networks. Utilizing meta-RL and policy distillation techniques, the framework is designed to adapt to diverse objective trade-offs quickly. Besides, some research \cite{Chang2020, Israr2023, Wu2021, Amine2022} investigate sleep mode as a load management strategy for base stations. Different from our paper, the above works do not consider the use of passively cooling.

For single PCBS, we have addressed online throughput maximization in \cite{YuZhYoYu24}. In the current paper, the system model consists of multiple PCBSs, and the resulting optimization problem is much more challenging because the PCBSs are interference-coupled. That is, the resource allocation in one PCBS cell directly affects the performance in the other cells, and thus optimization cannot be done for the individual PCBSs separately. In addition, the extended system model addresses the effect of MIMO. By modeling these aspects to enable to incorporate them in the learning framework, the current paper contains significant extensions of \cite{YuZhYoYu24}.


\section{System Model and Problem Formulation}

Consider a multi-cell OFDMA MIMO system using PCBSs. The set of cells is denoted by $\mathcal{I} = \{1,2,\ldots,I\}$, and the set of users represented by $\mathcal{J} = \{1,2,\ldots,J\}$. The subset of users in cell $i$ is denoted by $\mathcal{J}_i$.

\subsection{Load Coupling in Multi-cell MIMO Systems} \label{subsec:loadcoupling}

In the downlink, for a PCBS and a user, denote by $N_T$ and $N_R$ the numbers of transmit and receive antennas, respectively. The communication takes place over discrete resource blocks (RBs), each having bandwidth $B$. The wireless channel from PCBS $i$ to user $j$ is narrowband and time-invariant, and it is characterized by an $N_R\times N_T$ matrix $\mathbf{H}_{ij} \in \mathbb{C}^{N_R\times N_T}$. The transmitted symbol vector $\mathbf{x}_{ij} \in \mathbb{C}^{N_T \times 1}$ is comprised by $N_T$ independent input symbols $x_{ij,1}, x_{ij,2}, \ldots, x_{ij,N_T}$. The received signal at user $j$ from PCBS $i$, $\mathbf{y}_{ij} \in \mathbb{C}^{N_T\times 1}$, can be expressed in matrix form as
\begin{align}
    \mathbf{y}_{ij} = \underbrace{\vphantom{\sum_{\ell \in \mathcal{I}
    \setminus i} \sum_{k\in \mathcal{J}_\ell}
    \sqrt{\frac{E}{N_T}}}\sqrt{\frac{E}{N_T}}\mathbf{H}_{ij}\mathbf{x}_{ij}}_{signal}
    + \underbrace{\sum_{\ell \in \mathcal{I} \setminus i} \sum_{k\in
    \mathcal{J}_\ell}
    \sqrt{\frac{E}{N_T}}\mathbf{H}_{\ell j}\mathbf{x}_{\ell k}}_{interference}
    + \underbrace{\vphantom{\sum_{\ell \in \mathcal{I} \setminus i}
    \sum_{k\in \mathcal{J}_\ell}
    \sqrt{\frac{\rho_{\ell k}E}{N_T}}}\mathbf{z}_j}_{noise},
\end{align}
where $\mathbf{H}_{\ell j}$ and $\mathbf{x}_{\ell j}$ denote the channel matrix and transmitted symbol vector by BS $\ell$ to user $j$, respectively, $E$ signifies the power of transmitted signals, and $\mathbf{z}_{j} \in \mathbb{C}^{N_R\times 1}$ represents the noise vector, which is assumed to be zero-mean circular symmetric complex Gaussian (ZMCSCG). We assume that the transmission power for each transmit antenna is one unit, such that we have the following constraint,
\begin{equation}
    \text{Tr}(\mathbf{x}_{ij}\mathbf{x}_{ij}^{\rm T}) = N_T.
\end{equation}8

As commented earlier, the amount of RBs needed by one PCBS is coupled with that of the other PCBSs due to the inter-cell interference. Henceforth, we use the term load to refer to the amount of RB consumption of a cell. (Note that this entity is different from the data throughput.) Apparently, the load is directly proportional to how frequent transmissions occur.  Hence load is tightly connected to interference generation. At the two extremes, zero load means no transmission and hence no interference, and 100\% load means full occupation of all RBs, and hence another PCBS will receive inter-cell interference, no matter which RBs to use. To some extent, load represents the likelihood of interference. In so called load-coupling model, the load of an interfering cell is incorporated into the denominator of the signal-to-interference-plus-noise ratio (SINR), in form of a scaling parameter in front of interference. This interference modeling approach has been widely used in various context, see, e.g., \cite{majewski2010conservative, fehske2012aggregation, Siomina2012Analysis, Awan2020Robust, klessig2015performance, fehske2013concurrent, cavalcante2016low, arani2023haps}, and its accuracy has been demonstrated in~\cite{fehske2012aggregation,klessig2015performance}, among others. In the current paper, we extend the model to our MIMO setting.

We remark that the load of a cell is the sum of that for serving each of its users. For a user, the corresponding value is determined by the user throughput and SINR. The latter, by the load-coupling model, depends on the load levels of other cells. This is the reason why the cell load levels are coupled, and the relation is given by a non-linear system. Let $\rho_{\ell k} \in [0,1])$ represent the load of cell $\ell$ due to serving user $k$.  With the load-coupling model, the interference received at another user $j$ in another cell $i$ is given by
\begin{align}\label{interference}
    \Phi_{ij} = \sum_{\ell \in \mathcal{I} \setminus \{i\}} \sum_{k\in \mathcal{J}_\ell} \rho_{\ell k} \frac{E}{N_T}\|\mathbf{H}_{\ell j}\mathbf{x}_{\ell k}\|^2.
\end{align}
Intuitively, $\rho_{\ell k}$ reflects the likelihood that a user outside cell $\ell$ receives interference from cell $\ell$ due to that cell $\ell$ serves its user $k$. Therefore, the SINR for user $j$ in cell $i$ is given by
\begin{equation}\label{SINR}
    \text{SINR}_{ij} = \frac{\frac{E}{N_T}\|\mathbf{H}_{ij}\mathbf{x}_{ij}\|^2}{\Phi_{ij} + B N_0},
\end{equation}
where $N_0$ is the noise power spectral density. 

For user $j$ in cell $i$, the rate on one RB is  $B\log\left(1+ \text{SINR}_{ij} \right)$. Hence, to deliver throughput $D_{ij}$, we obtain the following 
\begin{equation}\label{eq:coupling}
    D_{ij} = \rho_{ij} B\log\left(1+ \text{SINR}_{ij} \right)
\end{equation} 
Note that the above system contains the transmit vectors to be designed, as well as the load values that are coupled with each other with non-linear relations. Additionally, let $D_i$ represent the total throughput of cell $i$, which is bounded by the capacity of the BBU:
\begin{align}
D_i \le D^{\max}.
\end{align}

\subsection{Thermal Dynamics of PCBS}\label{subsec:PCBS}

We consider slotted time, and denote the time horizon under consideration by time slot set $\mathcal{T} = \{1,2,\ldots,T\}$. The duration of each time slot, denoted by $\delta$, is flexible and typically on the order of seconds to facilitate effective thermal management.  For PCBS $i$, each time slot $t$ is characterized by the BBU chip temperature $\Psi^t_i$, updated at the beginning of the slot. Subsequently, the term ``temperature" (with unit Celsius) refers to the BBU chip temperature, and the ambient temperature is denoted by $\hat{\Psi}^t_i$.

The interplay between heat generation and heat dissipation governs the thermal dynamics of the PCBS. Let $P^{\uparrow t}_i$ denote the heat generated and $P^{\downarrow t}_i$ the heat dissipated in time slot $t$, measured in watts (W). Based on fundamental heat transfer principles \cite{BBUcooling}, the temperature evolution can be described by
\begin{equation}\label{5}
    \Psi^{t+1}_i = \Psi^{t}_i + \lambda_i \delta (P^{\uparrow t}_i - P^{\downarrow t}_i),
\end{equation}
where $\lambda_i$ (in $^\circ$C/Joule) represents the reciprocal of the thermal capacitance of the BBU of PCBS $i$. The generation of heat $P^{\uparrow t}_i$, is linked to the BBU's chip power consumption, comprising both dynamic power $P^{\zeta t}_i$ and static power $P^{\vartheta t}_i$ components \cite{dynsta}, in watts:
\begin{equation}
    P^{\uparrow t}_i = P^{\zeta t}_i + P^{\vartheta t}_i.
\end{equation}

Dynamic power consumption $P^{\zeta t}_i$ is directly proportional to the computational workload of PCBS $i$ in time slot $t$. Fo time slot $t$, we use $D_i^t = \sum_{j \in \mathcal{J}_i} D_{ij}^t$, where $D_{ij}^t$ is the throughput of user $j$, to represent the throughput of cell $i$.  For PCBS $i$, the dynamic power can be expressed as~\cite{dyn}
\begin{equation}
    P^{\zeta t}_i = \mu_i D_i^t,
\end{equation}
where $\mu_i$ (in Watt/bit) is a coefficient specific to the BBU of PCBS $i$. Conversely, static power consumption $P^{\vartheta t}_i$ exhibits an exponential relationship with the chip temperature \cite{dynsta,sta}, described by
\begin{equation}
    P^{\vartheta t}_i = \alpha_i {\rm e}^{\beta_i \Psi^{t}_i} + \gamma_i,
\end{equation}
with parameters $\alpha_i$, $\beta_i$, and $\gamma_i$ characterizing the chip's material properties.

Heat dissipation efficiency for time slot $t$, denoted by $\sigma^t_i > 0$, is influenced by the temperature difference between the chip and its surroundings. Following Newton's law of cooling~\cite{BBUcooling}, the dissipated power is given by
\begin{equation}\label{9}
    P^{\downarrow t}_i = \sigma^t_i (\Psi^{t}_i - \hat{\Psi}^t_i).
\end{equation}
Based on \eqref{5}-\eqref{9}, the temperature at the start of time slot $t+1$ (or the end of slot $t$) can be determined by:
\begin{align}\label{temp_max}
    &\Psi^{t+1}_i \notag  \\
    \!=\! &\max\left\{ \Psi^{t}_i \!+\! \lambda_i \delta \left[\mu_i D_i^t \! + \! \alpha_i {\rm e}^{\beta_i \Psi^{t}_i} \!+\! \gamma_i \!-\! \sigma^t_i (\Psi^{t}_i \!-\! \hat{\Psi}^t_i) \right], \hat{\Psi}^{t+1}_i\right\}.
\end{align}

We can see from \eqref{temp_max} that the BBU chip temperature does not fall below the ambient temperature. A critical operational constraint is to maintain the BBU chip temperature below a safe threshold, $\bar{\Psi}$, across all time slots,
\begin{equation}
    \Psi^{t}_i \le \bar{\Psi}, \   t \in \mathcal{T}.
\end{equation}

\pgfplotsset{width=0.5\textwidth, height=4.5cm, compat=1.9}

\begin{figure}
    \centering
    \begin{tikzpicture}
        \begin{axis}[
            xtick style={draw=none}, 
            ytick pos=left, 
            axis on top,
            tick label style={font=\scriptsize},
            xtick align=inside,
            ylabel={Temperature ($^\circ$C)},
            xticklabels=\empty, 
            xmin=0, xmax=10,
            ymin=0, ymax=130,
            xmajorgrids=true,
            xtick={0, 1, 2, 3, 4, 5, 6, 7, 8, 9, 10},
            ytick={40, 80, 120},
            grid style=densely dashed,
            tick label style={font=\scriptsize},
            label style={font=\small},
            ylabel style={align=center},
            legend style={
                at={(0.615, 0.025)},
                anchor=south west,
                legend columns=1,
                font=\scriptsize
            },
            minor x tick num=1,
            minor tick length=0,
            grid style={dashed},
        ]
            \addplot[ color=blue, mark size=1pt, mark=square, line width=0.8pt]     
            coordinates {
                (0, 25) 
                (1, 40) 
                (2, 70) 
                (3, 110) 
                (4, 180) 
            };
                
            \addplot[ color=orange, mark size=1pt, mark=triangle, mark options={solid}, line width=0.8pt]
            coordinates {
                (0, 25) 
                (1, 35) 
                (2, 50) 
                (3, 51) 
                (4, 53) 
                (5, 52) 
                (6, 49) 
                (7, 51) 
                (8, 52) 
                (9, 51) 
                (10, 55) 
            };

            \addplot[ color=green!70!black, mark size=1pt, mark=o, mark options={solid}, line width=0.8pt]     
            coordinates {
                (0, 25) 
                (1, 37) 
                (2, 60) 
                (3, 85) 
                (4, 116) 
                (5, 92) 
                (6, 80) 
                (7, 114) 
                (8, 91) 
                (9, 83) 
                (10, 117) 
            };
            
            \addplot[ color=gray, densely dashed, mark options={solid}, line width=0.8pt]
            coordinates {
                (0, 120) 
                (100, 120) 
            };
            
            \legend{Aggressive, Conservative, (Naive) Adaptive}
        \end{axis}
    \end{tikzpicture}

    \begin{tikzpicture}
        \begin{axis}[
            xtick style={draw=none}, 
            ytick pos=left, 
            ylabel shift=8.5pt, 
            xlabel={Time slot},
            ylabel={Throughput (Mbps)},
            xmin=0, xmax=10,
            ymin=0, ymax=3.5,
            ytick={0, 1, 2, 3},
            xtick={0, 1, 2, 3, 4, 5, 6, 7, 8, 9, 10},
            grid style=densely dashed,
            xmajorgrids=true,
            legend style={
                at={(0.61, 0.03)},
                anchor=south west,
                legend columns=1,
                font=\scriptsize
            },
            tick label style={font=\scriptsize},
            label style={font=\small},
            minor x tick num=1,
            minor tick length=0,
            grid style={dashed},
        ]
            \addplot[ color=blue, densely dotted, line width=1.2pt]     
            coordinates {
                (0, 0.0)
                (1, 1.0)
                (2, 1.6)
                (3, 2.1)
                (4, 2.1)
                (5, 2.1)
                (6, 2.1)
                (7, 2.1)
                (8, 2.1)
                (9, 2.1)
                (10, 2.1)
            };

            \addplot[ color=orange, densely dotted, line width=1.2pt]     
            coordinates {
                (0, 0.0)
                (1, 0.6)
                (2, 0.8)
                (3, 1.05)
                (4, 1.22)
                (5, 1.43)
                (6, 1.62)
                (7, 1.84)
                (8, 2.02)
                (9, 2.23)
                (10, 2.44)
            };
                
            \addplot[ color=green!70!black, densely dotted, line width=1.2pt]
            coordinates {
                (0, 0.0)
                (1, 0.88)
                (2, 1.44)
                (3, 1.88)
                (4, 2.28)
                (5, 2.28)
                (6, 2.28)
                (7, 2.85)
                (8, 2.85)
                (9, 2.85)
                (10, 3.42)
            };
            \legend{Aggressive, Conservative, (Naive) Adaptive}
        \end{axis}
    \end{tikzpicture}

    \caption{A BBU chip temperature evolution and the accumulated throughput under three different throughput policies.}\label{three_policy}
\end{figure}
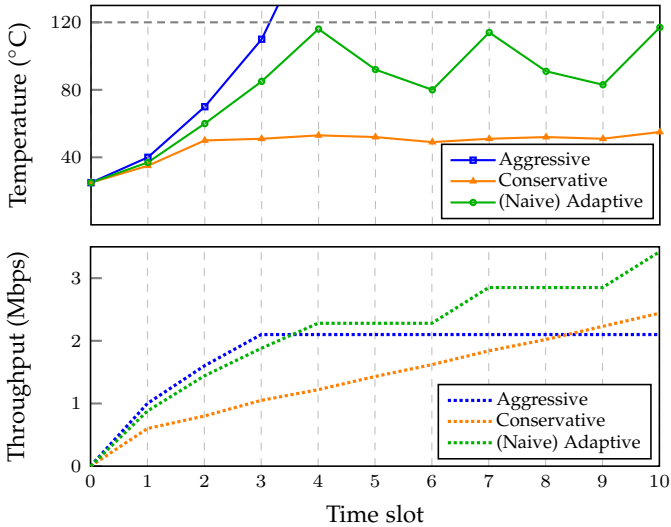

We can adjust the throughput delivered by the BBU in every time slot to control the chip temperature of a PCBS. We illustrate three throughput policies in Fig.~\ref{three_policy} as examples, to show the trade-off between throughput and temperature.  Here, the throughput represents the accumulated value.  For the illustration, $120^\circ$C is the maximum allowable temperature of the chip.

\begin{enumerate}
    \item {\bf Aggressive policy}: Under this policy, the BBU always attempts to handle the maximum possible workload. As a result, the BBU continuously generates heat exceeding the amount of heat dissipation, causing overheating. The throughput is higher than the two other policies at the beginning but cannot be sustained.
    \item {\bf Conservative policy}: With this policy, the BBU     delivers only a small amount of throughput, allowing the generated heat to dissipate within each time slot. This keeps the chip temperature around $50^\circ$C. However, as can be seen, the conservative action results in rather low throughput.
    \item {\bf (Naive) Adaptive policy}: This policy adapts the     throughput to the temperature evolution.  The temperature curve is in fact what we would normally expect from thermal management. However, note that the policy is still somewhat naive because it periodically sets the throughput to be zero for some time slots (shown by the plateaus in the accumulated throughput curve) to allow BBU to cool down.
\end{enumerate}

\begin{remark}
    The adaptive policy may be seemingly easy to design. This is however not the case for two reasons.  First, the heat dissipation efficiency of future time slots is not known in an online setting. Second, the throughput of the cells, as discussed earlier, cannot be set separately, because they together determine the radio resource consumption, i.e., load, and we need to ensure this consumption does not exceed what is available.
\end{remark}

\subsection{Problem Formulation}
For problem formulation of throughput maximization, we now denote explicitly the SINR$_{ij}$ as function of $\mathbf{x}_{ij}^t$ and interference $\Phi_{ij}^t$, as $\text{SINR}_{ij}^{t}(\mathbf{x}_{ij}^t, \Phi_{ij}^t)$. Note that $\Phi_{ij}^t$ is a function $\Phi_{ij}^t$ depends on $\mathbf{x}_{\ell k}^t$ and the load values $\rho_{\ell k}^{t}$, $k \in \mathcal{J}_\ell, \ell \in \mathcal{I} \setminus \{i\}$. To ensure fairness, we impose a constraint that in each time slot, the throughput of every user must not fall below a certain threshold. This threshold is defined as the total throughput of the user's serving cell divided by the number of users in the cell. Thus, we have the following constraint:
\begin{align}
    &\rho_{ij}^t B \log\left[1 \!+\! \text{SINR}_{ij}^{t}(\mathbf{x}_{ij}^t, \Phi_{ij}^{t})\right] \ge \frac{D_i^t}{|\mathcal{J}_i|}, j \in \mathcal{J}_i, i \in \mathcal{I}, t \in \mathcal{T},
\end{align}
where $\rho_{ij}^t$ is the the proportion of time-frequency resources allocated to user $j$ in cell $i$ in time slot $t$.

We use $\rho_{i}^t$ and $\Psi^{t}_i$ to represent the level (or proportion) of time-frequency resources used at PCBS $i$ and the BBU chip temperature of PCBS $i$ in time slot $t$, respectively. As mentioned earlier, our problem is of online decision-making. For the sake of presentation, however, let us for a moment assume that the heat dissipation coefficients of all time slots are known in advance.  The problem can then be formulated as follows.
\begin{subequations} \label{formulation:T}
    \begin{align}
        \max \quad & \sum_{t \in \mathcal{T}} \sum_{i \in \mathcal{I}} D_i^t \label{obj}\\
        \text{o.v.} \ \quad & D_i^t \in [0, D^{\max}], \mathbf{x}_{ij}^t, \rho_{ij}^t, \rho_{i}^t, \Psi^{t}_i,   i \in \mathcal{I}, j \in \mathcal{J}_i, t \in \mathcal{T} \notag\\
        \text{s.t.} \ \quad & \Psi^{t+1}_i \ge \Psi^{t}_i \!+\! \lambda_i \delta \!\left[\!\mu_i D_i^t \!+\! \alpha_i {\rm e}^{\beta_i \Psi^{t}_i} \!+\! \gamma_i \!-\! \sigma^t_i (\Psi^{t}_i \!-\! \hat{\Psi}^t_i) \!\right]\!, \notag \\ 
        &\qquad\qquad\qquad\qquad\qquad\qquad   i \in \mathcal{I}, t \in \mathcal{T} \setminus \{T\} \label{T_C1}\\
        & \hat{\Psi}^t_i \le \Psi^{t}_i \le \bar{\Psi} ,   i \in \mathcal{I}, t \in \mathcal{T} \label{T_C2}\\
        & \rho_{ij}^t B \log\left[1+ \text{SINR}_{ij}^{t}(\mathbf{x}_{ij}^t, \Phi_{ij}^{t})\right] \ge \frac{D_i^t}{|\mathcal{J}_i|}, \notag \\ 
        &\qquad \qquad \qquad \qquad \qquad \quad   j \in \mathcal{J}_i, i \in \mathcal{I}, t \in \mathcal{T} \label{LC_C1}\\
        & \text{Tr}(\mathbf{x}_{ij}^t{\mathbf{x}_{ij}^t}^{\rm T}) = N_T,   j \in \mathcal{J}_i,   i \in \mathcal{I} \\
        & \rho_i^t = \sum_{j\in \mathcal{J}_i}\rho_{ij}^t \le \bar{\rho},   i \in \mathcal{I}, t \in \mathcal{T}\label{LC_C2}
    \end{align}
\end{subequations}

The objective \eqref{obj} aims to maximize the system's total throughput over all time slots. Constraints \eqref{T_C1} and \eqref{T_C2} ensure that the temperature evolution complies with passive cooling dynamics and stays beneath the upper limit. Moreover, constraint \eqref{LC_C2} states that resulting time-frequency resource consumption, induced by the nonlinear system \eqref{LC_C1}, adheres to the resource limit $\bar{\rho}$ in all cells.


\section{Algorithm Design: Overview}

The problem in \eqref{formulation:T} presents challenges due to the cells' interdependence in load and the non-convex nature of \eqref{LC_C1}. Moreover, for online decision-making, we do not have the parameter values of all time slots in order to have \eqref{formulation:T} in its complete form. For these reasons, we apply the reinforcement learning (RL) method for problen-solving.

The goal of the algorithm design of RL is to provide a policy enabling to schedule throughput effectively across various states by maximizing the long-term cumulative reward. To this end, we model the multi-cell system as an agent, and the action of the agent consists of the variables to be optimized.  The fundamental components of this agent are as follows.

\subsection{State and Action}\label{subsec:State}

Based on \eqref{T_C1}-\eqref{T_C2}, for a given time slot $t$, the agent needs access the current system specifics, which encompass \begin{itemize}
    \item Ambient temperature, $\hat{\Psi}^t_i,   i \in \mathcal{I}$,
    \item Chip temperature, ${\Psi}^t_i,   i \in \mathcal{I}$,
    \item Optionally, the heat dissipation efficiency, $\sigma_i^t,   i \in \mathcal{I}$.
\end{itemize}

Note that if a time slot is short, then it is reasonable to assume that the heat dissipation efficiency at the very beginning of a time slot remains for the entire slot duration; otherwise, we may not assume the knowledge of heat dissipation for the current time slot. Therefore we consider two scenarios:
\begin{enumerate}
    \item Informed-Heat-Dissipation (IHD) Scenario: In this case, PCBS knows the heat dissipation efficiency for the current time slot. Thus, the state $s_t$ in the time slot $t$ is expressed by
    \begin{align}
    s_t = \left[\left(\hat{\Psi}^t_i, {\Psi}^t_i, \sigma_i^t\right),   i \in \mathcal{I}\right].
    \end{align}
    \item Uninformed-Heat-Dissipation (UHD) Scenario: In contrast, the heat dissipation efficiency is not known, so we utilize $s_i'$ to delineate its state as:
    \begin{align}
    s_t' = \left[\left(\hat{\Psi}^t_i, {\Psi}^t_i \right),   i \in \mathcal{I}\right].
    \end{align}
\end{enumerate}
With state $s_t$ in time slot $t$, the agent takes an action in action space. The action $a_t$ is defined as
\begin{align}
    a_t =  \left[D_i^t,   i \in \mathcal{I}\right].
\end{align}

\subsection{Reward Mechanism}\label{subsec:Reward}

We employ soft actor-critic (SAC) to refine policies with the aim of maximizing rewards in the environment \cite{haarnoja2018soft}. SAC finds the maximum of the expected sum of rewards and takes the expected entropy objective to adopt stochastic policies. The maximum entropy objective function is defined as:
\begin{align}
    J(\pi) = \sum_{t=0}^{T} \mathbb{E}_{(s_t,a_t) \sim \rho_\pi} [R(s_t, a_t) + k H(\pi( \cdot | s_t))],
\end{align}
where $R(s_t, a_t)$ is the reward function, $k$ is a factor to determine the importance of entropy relative to the reward, and $\pi(\cdot | s_t)$ is the probability distribution of subsequent actions given the state $s_t$, typically resorting to a Gaussian distribution. Additionally, $H(\pi(\cdot | s_t))$ embodies the entropy component, represented as $H(\pi(\cdot | s_i)) \triangleq \mathbb{E}_a[-\log(\pi(a | s_t))]$. It is important to highlight that the agent actively tries various actions to maximize the target entropy. This strategy improves the exploratory feature of SAC.

Our design of the reward mechanism is comprised of two elements.
\begin{enumerate}
    \item Design for resource constraint \eqref{LC_C2}: Should the agent initiate an action to set the throughput levels such that the resulting load in any cell exceeds resource availability, the action will not be executed, leading to a reward of zero for the agent. Conversely, if the action is within the resource limit, the agent may be given a positive reward.
    \item Design for temperature constraint \eqref{T_C2}: In the event that an action proposed by the agent leads to the estimated CPU temperature exceeding the maximum allowed threshold, the action will be rejected, and the agent will receive no reward. On the other hand, actions that maintain the BBU temperature within safe limits will be rewarded with a positive value. We will present a refinement of this basic design later on.
\end{enumerate}


\section{Reward Design for Resource Limit}

\subsection{Preliminaries}
In the context of RL, the resource constraint \eqref{LC_C2} boils down to the following: For an action $\left\{D_i^t|   i \in \mathcal{I}\right\}$, can these throughput levels be met within the resource limit?  This amounts to the following min-max problem. For simplicity, the time slot index is omitted.
\begin{subequations}\label{multi-cell problem}
    \begin{align}
        \min \quad & \hat{\rho} \label{sing-cell problem obj} \\
        \text{o.v.} \quad & \hat{\rho}, \mathbf{x}_{ij}, \rho_{ij},    i \in \mathcal{I}, j \in \mathcal{J}_i \notag\\
        \text{s.t.} \quad & \rho_{ij} B \log\left[ 1\!+\!\text{SINR}_{ij} (\mathbf{x}_{ij}, \Phi_{ij})\right] \ge \frac{D_i}{|\mathcal{J}_i|},   j \in \mathcal{J}_i,   i \in \mathcal{I} \label{multi-cell problem c1} \\
        & \text{Tr}(\mathbf{x}_{ij}\mathbf{x}_{ij}^{\rm T}) = N_T,   j \in \mathcal{J}_i,   i \in \mathcal{I}\label{multi-cell problem c2}\\
        & \rho_i = \sum_{j\in \mathcal{J}_i}\rho_{ij} \le \hat{\rho} ,   i \in \mathcal{I} \label{multicelllast} 
    \end{align}
\end{subequations}

In this problem, $\hat{\rho}$ is an auxiliary variable, and by \eqref{multicelllast}, it takes the value of maximum load over all cells. After solving the problem, if the optimal $\hat{\rho}$ is less than or equal to $\bar{\rho}$, action $\left\{D_i^t|   i \in \mathcal{I}\right\}$ is feasible.

The multi-cell problem \eqref{multi-cell problem} concerns a single time slot. However, we have to deal with the cells' interdependence in interference and the resulting non-convexity.  To proceed, we consider local optimization of the symbol vector within each cell, that is, each cell finds the solution that maximizes the SINR of each of its users (even though the solution does not mean minimum interference to other cell users). This is a reasonable trade-off, because the solution is locally optimal within each cell, and maximum SINR translates into minimum resource consumption, which reduces interference to other cells.  Moreover, the computation becomes confined within each individual cell.

Consider cell $i$ and user $j \in \mathcal{J}_i$. To maximize the SINR by optimizing $\mathbf{x}_{ij}$, the problem is formulated as
\begin{subequations}\label{user problem}
    \begin{align}
        \max \quad & \|\mathbf{H}_{ij}\mathbf{x}_{ij}\|^2 \\
        \text{o.v.} \quad & \mathbf{x}_{ij} \notag \\
        \text{s.t.} \quad & \text{Tr}(\mathbf{x}_{ij}\mathbf{x}_{ij}^{\rm T}) = N_T
    \end{align}
\end{subequations}
To solve \eqref{user problem}, we use singular value decomposition (SVD) to deal with the channel matrix, i.e., $\mathbf{H}_{ij} = \mathbf{U}_{ij} \mathbf{\Sigma}_{ij} \mathbf{V}^{\rm T}_{ij}$, creating multiple independent virtual SISO subchannels. Subsequently, we employ the well-known water-filling algorithm \cite{knopp1995information} for $\mathbf{\Sigma}_{ij}$ to determine the optimal signal vector $\mathbf{x}_{ij}^*$. Then, the PCBS uses right singular matrix $\mathbf{V}_{ij}$ for precoding, transforming the signal into the subchannels. The user uses left singular matrix $\mathbf{U}_{ij}$ for post-processing, decoupling the received signal into independent subchannel signals.

\subsection{Multi-cell Optimization and Reward Design}\label{sec:noar}

Having computed $\mathbf{x}^*_{ij}, j \in \mathcal{J}_i, i \in \mathcal{I}$, the problem becomes more tractable. More specifically, the interference term $\Phi_{ij}$ takes the following form
\begin{align}\label{ninterference}
    \Phi_{ij} = \frac{E}{N_T}\sum_{\ell \in \mathcal{I} \setminus \{i\}} \sum_{k\in \mathcal{J}_\ell} \rho_{\ell k} h_{\ell kj},
\end{align}
where $h_{\ell kj} = ||\mathbf{H}_{\ell k} \mathbf{x}^*_{\ell k}||^2$. In addition, we define $h_{ij} = ||\mathbf{H}_{ij} \mathbf{x}^*_{ij}||^2$. Plugging in the above into \eqref{SINR}, and let $N_c = \frac{BN_0 N_T}{E}$, the problem reads
\begin{subequations}\label{mproblem}
    \begin{align}
        \min \quad & \hat{\rho} \label{mproblemobj} \\
        \text{o.v.} \quad & \hat{\rho}, \rho_{ij},    i \in \mathcal{I}, j \in \mathcal{J}_i \notag\\
        \text{s.t.} \quad &  \rho_{ij} B \log\left[ 1+ \frac{h_{ij}}{ \sum\limits_{\ell \in \mathcal{I} \setminus \{i\}} \sum\limits_{k \in \mathcal{J}_\ell} \rho_{\ell k} h_{\ell k j} + N_c} \right]\geq \frac{D_i}{|\mathcal{J}_i|}, \nonumber \\
        & \qquad \qquad \quad \qquad \qquad \qquad \qquad   j \in \mathcal{J}_i,   i \in \mathcal{I} \label{mproblemsinr} \\
        & \rho_i = \sum_{j\in \mathcal{J}_i}\rho_{ij} \le \hat{\rho} ,   i \in \mathcal{I} \label{mproblemlast} 
    \end{align}
\end{subequations}

\begin{lemma}
\label{lemma:opt}
There exists at least one optimal solution of \eqref{mproblem}, such that
\eqref{mproblemsinr} holds as equality for all users.
\end{lemma}
\begin{proof}
Denote by $\mathbf{\rho}^*$ an optimal solution, and suppose for this solution
\eqref{mproblemsinr} holds as strict inequality for user $j \in \mathcal{J}_i$. 
Clearly, we can reduce $\rho^*_{ij}$ slightly without violating
\eqref{mproblemsinr} for this user.  Moreover, a smaller value of
$\rho_{ij}$ implies less interference (i.e., smaller value of the
SINR denominator for any user in other cells), and hence
\eqref{mproblemsinr} remains to hold for all the other users, and the
lemma follows.
\end{proof}

Define $\boldsymbol{\rho}_{-ij} = [\rho_{\ell k} | \ell \in
\mathcal{J} \setminus \{i\}, k \in \mathcal{J}_\ell]$.  This is the vector
of all user-level load (i.e., resource consumption) for meeting user
throughput in cells other than $i$.  Moreover, let
$f_{ij}(\boldsymbol{\rho}_{-ij})$ denote the value of $\rho_{ij}$ as
function of $\boldsymbol{\rho}_{-ij}$, that is,
\begin{align}\label{eq:rhofunction}
f_{ij}(\boldsymbol{\rho}_{-ij}) = \frac{D_i}{B |\mathcal{J}_i| \log \left[ 1+ \frac{h_{ij}}{ \sum\limits_{\ell \in \mathcal{I} \setminus \{i\}} \sum\limits_{k \in \mathcal{J}_\ell} \rho_{\ell k} h_{\ell k j} + N_c} \right]}.
\end{align}

By Lemma~\ref{lemma:opt}, what we are seeking is the solution to the following system.
\begin{align}\label{eq:rhosystem}
    \rho_{ij} = f_{ij}(\boldsymbol{\rho}_{-ij})
\end{align}

\begin{lemma}\label{lemma:sif}
    Function $f_{ij}(\boldsymbol{\rho}_{-ij}), j \in \mathcal{J}_i, i \in \mathcal{I}$ is a standard interference function (SIF).
\end{lemma}
\begin{proof}
    See Appendix~\ref{sec:proofsif}.
\end{proof}

\begin{theorem}\label{theorem:sif}
    Let $\rho_{ij}^{(0)} \geq 0~ (j \in \mathcal{J}_i, i \in \mathcal{I})$ be an arbitrary non-negative initial point. Then, the fixed-point method, where $\rho_{ij}^{(\tau+1)} = f_{ij}(\boldsymbol{\rho}_{-ij}^{(\tau)}), \tau = 0,1 ,\dots$, converges to a unique point that is optimum of \eqref{mproblem}. 
\end{theorem}
\begin{proof}
    The fixed-point iterations converge to a unique point follows from the basic properties of SIF~\cite[Theorems 1 and 2]{yates1995framework}. This together with Lemma~\ref{lemma:opt} establish the theorem.
\end{proof}

\begin{algorithm}[tbp]
\caption{Fixed point method} \label{al:fixed_point}
\KwIn{ $\{D_i, \mathbf{x}_{ij}^*,   i \in \mathcal{I},   j \!\in \mathcal{J}\}, \varepsilon$;} 
\KwOut{$\hat{\rho}$;}   
Initialize $\boldsymbol{\rho}^{(0)}$ \label{step1}\;
$\tau \leftarrow 0$\;
\Repeat
{$||\boldsymbol{\rho}^{(\tau)}- \boldsymbol{\rho}^{(\tau-1)}||_{\infty} \leq \varepsilon$
}
{
$\tau \leftarrow \tau+1$\;
\For{{\rm cell} $i$, $  i \in \mathcal{I}$\label{parallel1}}{
    \For{{\rm user} $i$, $  j \in \mathcal{J}_i$}{
        $\rho_{ij}^{\tau} \leftarrow f_{ij}(\boldsymbol{\rho}_{-ij}^{(\tau-1)})$ 
    }
}
}
$\hat{\rho} = \max_{i \in \mathcal{I}} \sum_{j \in \mathcal{J}_i} \rho_{ij}^{(\tau)}$\;
\Return {$\hat{\rho}$}\;
\end{algorithm}

The fixed-point method for solving \eqref{mproblem} is given in Algorithm~\ref{al:fixed_point}. It is noted that when using \eqref{eq:rhofunction} to make updates, the complexity of interference computation is proportional to the number of users of the entire system. The complexity is greatly reduced if the interference computation involves cell-level information only. To this end, we present two interference approximations.

\begin{enumerate}
    \item Long range approximation: Assuming significant inter-cell distance, resulting in low diversity in the MIMO channel, $\mathbf{H}_{\ell j}$ is approximated as a rank-1 matrix.  Interference from cell $\ell$ at user $j$ in cell $i$ is approximated by $\rho_{\ell} E |g_{\ell j}|^2$, where $\rho_{\ell}$ is the cell-level load of $\ell$ and $g_{\ell j}$ is the singular value of the MIMO channel obtained through SVD, with $\mathbf{H}_{\ell j} = \mathbf{u} g_{\ell j} \mathbf{v}^{\rm T}$. Here, $\mathbf{u}$ and $\mathbf{v}$ are unitary orthogonal vectors representing the receiving and transmitting spaces, respectively.
    \item Upper bound approximation: For $j \in \mathcal{J}_i$, an upper bound of interference from cell $\ell$ can be obtained by the SVD and water-filling algorithm applied to \eqref{upperboundappro}.
    \begin{subequations} \label{upperboundappro} \begin{align} \max \quad
    & \|\mathbf{H}_{\ell j}\mathbf{x}_{\ell j}\|^2 \\ \text{o.v.} \quad &
    \mathbf{x}_{\ell j} \notag \\ \text{s.t.} \quad &
    \text{Tr}(\mathbf{x}_{\ell j}\mathbf{x}_{\ell j}^{\rm T}) = N_T \end{align}
    \end{subequations} 
    
    Let $\|\mathbf{H}_{\ell j}\mathbf{x}_{\ell j}^\bigstar \|^2$ be the maximum value of \eqref{upperboundappro}. Clearly,     $\rho_\ell \frac{E}{N_T} \|\mathbf{H}_{lj}\mathbf{x}_{\ell j}^\bigstar \|^2$ is an upper bound of interference. Therefore, $\sum_{\ell \in \mathcal{I}, \ell \neq i} \rho_{\ell} \frac{E}{N_T} \|\mathbf{H}_{\ell j}\mathbf{x}^\bigstar _{\ell j}\|^2$    represents an upper bound of the total interference received by user $j$.
\end{enumerate}
By either of the approximations, the interference experienced by user $j \in \mathcal{J}_i$ due to another cell $\ell$ is a function in the cell-level load $\rho_\ell$. Letting $\boldsymbol{\rho}_{-i} = [\rho_1, \rho_2,...,\rho_{i-1}, \rho_{i+1},...,\rho_I]$, the resulting system equation is $\rho_i = f_i(\boldsymbol{\rho}_{-i}) = \sum_{j \in \mathcal{J}_i} f_{ij}(\boldsymbol{\rho}_{-i}), i \in \mathcal{I}$.

\begin{theorem}
    Function $f_i(\boldsymbol{\rho}_{-i})$ for either of the approximations is an SIF.
\end{theorem}
\begin{proof}
    See Appendix~\ref{sec:proofasif}.
\end{proof}

After obtaining $\hat{\rho}$ by fixed-point method for solving \eqref{mproblem} or any of the aproximations, we verify whether or not it is below the resource limit $\bar{\rho}$, i.e.,
\begin{align}\label{rho_condition}
    \hat{\rho} \le \bar{\rho}.
\end{align}
If the condition is met, the action, namely the throughput $D_i, i \in \mathcal{I}$, is admitted.  In this case, the reward with respect to cell $i$ equals $D_i$.  Otherwise, the action is denied and the reward becomes zero. That is,
\begin{align}\label{tem_reward}
    R_{\rho_i}(s_t, a_t) = \left\{
                        \begin{array}{lcl}
                        D_i &      & {\text{if }\eqref{rho_condition}\text{ holds}}\\
                        0   &      & {\text{otherwise}}
                        \end{array} \right.
\end{align}


\section{Reward Design for Temperature Limit}\label{sec:temp}

The exploration phase in RL comes with the risk of overheating. We design a reward system to make temperature stay within the limit while providing the agent with accurate feedback. The policy zone, denoted by $\Pi$, is defined as:
\begin{align}
    \Pi = \Big\{\pi | \Psi_i^t \!+\! \lambda_i \delta \mathbb{E}_{\sigma_i^t, D_i^t \sim \pi (\cdot| s_t )} \Big[ \mu_i D_i^t \!+\! \alpha_i {\rm e}^{\beta_i \Psi^{t}_i} \!+\! \gamma_i \! \notag\\
    -\! \sigma_i^t (\Psi_i^t - \hat{\Psi}_i^t )\Big]\le \bar{\Psi},   i \in \mathcal{I}, t \in \mathcal{T} \Big\}.
\end{align}
In addition, we introduce a temperature parameter $\Psi_i^{t, \text{risk}}$, of which the value shall be as large as possible, subject to the following condition. By the condition, at temperature $\Psi_i^{t, \text{risk}}$, the expected temperature will not exceed the maximum allowed value in the next time slot, even if the agent takes the action of maximum throughput.
\begin{align}\label{riskcondition}
    &\Psi_i^{t, \text{risk}} + \lambda_i \delta \mathbb{E}_{\sigma_i^t} \left[ \mu_i D_{\max} +\! \alpha_i {\rm e}^{\beta_i \Psi^{t}_i} \!+\! \gamma_i \!-\! \sigma_i^t (\Psi_i^t - \hat{\Psi}_i^t )\right]\le \bar{\Psi}, \notag\\
    &\qquad \qquad \qquad \qquad \qquad \qquad \qquad \qquad \qquad    t \in \mathcal{T}, i \in \mathcal{I}
\end{align}

\begin{figure}[ht]
	\begin{center}

        \tikzset{every picture/.style={line width=0.75pt}} 
        
        \begin{tikzpicture}[x=0.75pt,y=0.75pt,yscale=-1,xscale=1]
        
        \draw    (199.97,265.5) -- (199.97,94.24) ;
        \draw [shift={(199.97,92.24)}, rotate = 90] [color={rgb, 255:red, 0; green, 0; blue, 0 }  ][line width=0.75]    (10.93,-3.29) .. controls (6.95,-1.4) and (3.31,-0.3) .. (0,0) .. controls (3.31,0.3) and (6.95,1.4) .. (10.93,3.29)   ;
        \draw   (195.43,162.58) .. controls (190.76,162.58) and (188.43,164.91) .. (188.43,169.58) -- (188.43,204.21) .. controls (188.43,210.88) and (186.1,214.21) .. (181.43,214.21) .. controls (186.1,214.21) and (188.43,217.54) .. (188.43,224.21)(188.43,221.21) -- (188.43,258.83) .. controls (188.43,263.5) and (190.76,265.83) .. (195.43,265.83) ;
        \draw  [fill={rgb, 255:red, 255; green, 136; blue, 136 }  ,fill opacity=1 ] (199.97,117) -- (358.22,117) -- (358.22,161.83) -- (199.97,161.83) -- cycle ;
        \draw  [fill={rgb, 255:red, 254; green, 244; blue, 126 }  ,fill opacity=1 ] (199.97,161.83) -- (358.22,161.83) -- (358.22,216.17) -- (199.97,216.17) -- cycle ;
        \draw  [fill={rgb, 255:red, 199; green, 255; blue, 142 }  ,fill opacity=1 ] (199.97,216.17) -- (358.22,216.17) -- (358.22,265.5) -- (199.97,265.5) -- cycle ;
        
        \draw (75.35,198) node [anchor=north west][inner sep=0.75pt]    {$\text{Temprature range}$};
        \draw (95.35,216) node [anchor=north west][inner sep=0.75pt]    {$\text{of Zone}~ \Pi $};
        \draw (360,202.82) node [anchor=north west][inner sep=0.75pt]    {$\Psi _{i}^{t,\text{risk}}$};
        \draw (360,152.9) node [anchor=north west][inner sep=0.75pt]    {$\bar{\Psi }$};
        \draw (120.23,90) node [anchor=north west][inner sep=0.75pt]   [align=left] {Temperature};
        \draw (210,125) node [anchor=north west][inner sep=0.75pt]   [align=left] {Deny the action \\and give penalty};
        \draw (210,175) node [anchor=north west][inner sep=0.75pt]   [align=left] {Permit the action, give\\ both reward and penalty};
        \draw (210,225) node [anchor=north west][inner sep=0.75pt]   [align=left] {Permit the action \\and give reward};

        \end{tikzpicture}
        
	\end{center}
	\caption{An illustration of the reward mechanism.}\label{fig:temperature}
\end{figure}
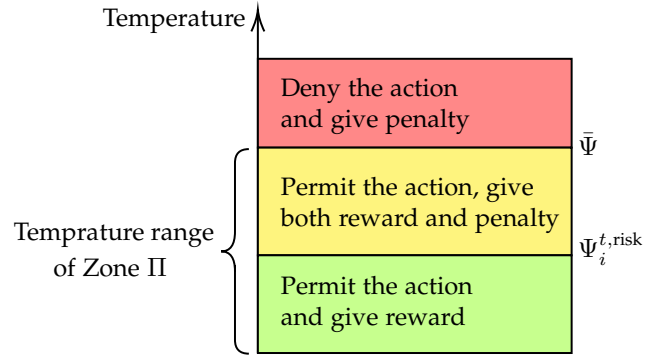

In our reward mechanism, we operate across the policies and let the system react with different rewards and decisions. Fig.~\ref{fig:temperature} illustrates the temperature range of policy zone $\Pi$ and the proposed reward structure. Specifically, for an action of time slot $t$, the system assesses the BBU chip temperature evolution, namely $\Psi_i^{t+1}$ and, responds as follows:
\begin{itemize}
    \item {$\Psi_i^{t+1} > \bar{\Psi}$}: The action is denied to prevent overheating, and a penalty (i.e., negative reward) is set.
    \item {$\Psi_i^{t, \text{risk}} < \Psi_i^{t+1} \leq \bar{\Psi}$}: In this case, there is no overheating but the condition \eqref{riskcondition} is not met. The action is permitted, with a throughput reward that is, however, reduced (i.e., a penalty term is added).
    \item {$\Psi_i^{t+1} \leq \Psi_i^{t, \text{risk}}$}: In this case, the temperature is in the safe range. The action is permitted, and the agent receives a positive reward for the throughput.
\end{itemize}
This mechanism ensures that the agent can learn to balance between the performance objectives in thermal management and throughput. However, risk condition \eqref{riskcondition} is implicit to the agent when we implement SAC, as computing the expectations above requires knowing $\pi$, and $\pi$ is not learned yet. We need to compute an explicit value of $\Psi_i^{t, \text{risk}}$. For time slot $t$ and cell $i$, we use $\Psi$ to represent the temperature, and this temperature can be represented using an optimization problem as follows.
\begin{subequations}\label{t_warning}
\begin{align}
    &\Psi_{i}^{t, \text{risk}} = \max_{\pi} \Psi \\ 
    &\text{s.t. } \Psi + \lambda_i \delta \mathbb{E}_{\sigma_i^t} \big[ \mu_i D_{i}^t +\! \alpha_i {\rm e}^{\beta_i \Psi} \!+\! \gamma_i \!-\! \sigma_i^t (\Psi - \hat{\Psi}_i^t )\big]\le \bar{\Psi}\label{warningtemperature}
\end{align}
\end{subequations}

Constraint \eqref{warningtemperature} states if $D_i^t$ is applied, then the resulting temperature $\Psi_{i}^{t+1}$ in the next time slot may not exceed the limit $\bar{\Psi}$. Solving the optimization problem exactly requires knowledge of the distribution of the heat dissipation efficiency as well as the learned policy distribution, which are not available a priori. We use the following approximation to solve problem \eqref{t_warning}:
\begin{align}
    &\Psi \!+\! \lambda_i \delta \mathbb{E}_{\sigma_i^t, D_i^t \sim \pi (\cdot| s_t )} \left[ \mu_i D_i^t \!+\! \alpha_i {\rm e}^{\beta_i \Psi} \!+\! \gamma_i \! -\! \sigma_i^t (\Psi - \hat{\Psi}_i^t )\right] \notag\\
    \le &\Psi \!+\!\lambda_i \delta \mathbb{E}_{\sigma_i^t}\left[ \mu_i D_{\max}\!+\! \alpha_i {\rm e}^{\beta_i \Psi} \!+\! \gamma_i \! -\! \sigma_i^t (\Psi - \hat{\Psi}_i^t )\right]\notag\\
    \approx & \Psi \!+\lambda_i \delta \left[ \mu_i D_{\max} \!+\! \alpha_i {\rm e}^{\beta_i \Psi} \!+\! \gamma_i \! -\! \bar{\sigma}_i^t (\Psi - \hat{\Psi}_i^t )\right]
\end{align}
where $\bar{\sigma}_i^t$ represents the estimated heat dissipation efficiency for time slot $t$ and cell $i$. Several methods are available for doing the estimation. For instance, one might employ historical average values or resort to the most conservative estimates, such as the worst value recorded in history. Hence, problem \eqref{t_warning} can be approximated as
\begin{subequations}\label{approximatedproblem}
\begin{align}
    &\Psi_{i}^{t, \text{risk}} = \max_{\pi} \Psi \\ 
    &\text{s.t. } \Psi \!+\lambda_i \delta \left[ \mu_i D_{\max} \!+\! \alpha_i {\rm e}^{\beta_i \Psi} \!+\! \gamma_i \! -\! \bar{\sigma}_i^t (\Psi - \hat{\Psi}_i^t )\right] \le \bar{\Psi}
\end{align}
\end{subequations}
Note that \eqref{approximatedproblem} can be efficiently solved using a bi-section search. Next, with the derived $\Psi_{i}^{t, \text{risk}}$, we use the reward function for cell $i$ as follows.
\begin{align}\label{eq:reward}
    &R_{\Psi_i}(s_t, a_t) =  \notag\\ 
    &\left\{
        \begin{array}{lcl}
        R_{\rho_i}(s_t, a_t) & & {\mathring{\Psi}_i^{t+1} <  \Psi^{t, \text{\rm risk}}_i}\\
        R_{\rho_i}(s_t, a_t)  \!+\! \Psi_{i}^{t, \text{risk}} \!-\! \mathring{\Psi}_i^{t+1} & & {\Psi^{t, \text{\rm risk}}_i \!\le\! \mathring{\Psi}_i^{t+1} \!<\!  \bar{\Psi}}\\
        \mathring{\Psi}_i^{t+1} - \bar{\Psi} & &\mathring{\Psi}_i^{t+1} \ge  \bar{\Psi}
        \end{array} 
    \right.
\end{align}
where $\mathring{\Psi}_i^{t+1}$ is an estimated temperature. The reward function $R_{\Psi_i}(s_t, a_t)$ described above is a piecewise function based on the estimated temperature $\mathring{\Psi}_i^{t+1}$, $\Psi^{t, \text{risk}}_i$, and the limit $\bar{\Psi}$. The decision made by \eqref{eq:reward} follows the general idea outlined earlier. There are three additional specifics, however. First, the reward (i.e., $R_{\rho_i}(s_t, a_t)$) is positive only if the action meets the resource limit (cf.~\eqref{tem_reward}).  Second, since we cannot obtain exact value of ${\Psi}_i^{t+1}$, estimated temperature value $\mathring{\Psi}_i^{t+1}$ is used in lieu.  Third, the penalty is set to be how much the estimated temperature exceeds the risk level $\Psi_{i}^{t, \text{risk}}$ or the limit $\bar{\Psi}$.

The flow of the entire reward mechanism, with respect to both the resource and temperature limits, is shown in Algorithm \ref{al:mechanism}, Note that for the UHD scenario, we use the estimated value $\bar {\sigma}_i^t$ as the input of $\sigma_i^t$.

\begin{algorithm}[tbp]
    \caption{Reward mechanism in time slot $t$} \label{al:mechanism}
    \KwIn{$\{D_i^t, \hat{\Psi}_i^t, \sigma_i^t,  i \in \mathcal{I}\}$} 
    \KwOut{$\{D_i^t, {\Psi}_i^{t+1},   i \in \mathcal{I}\}$, $R_t$} 
    Calculate $\hat{\rho} $ by Algorithm \ref{al:fixed_point}\\
    \If{$\hat{\rho} > \bar{\rho}$}{
        \For{{\rm cell} $i$, $  i \in \mathcal{I}$}{
            $ D_i^t \leftarrow 0$
        }
    }
    \For{{\rm cell} $i$, $  i \in \mathcal{I}$}{
        Obtain $\Psi_{i}^{t, \text{risk}}$ by solving \eqref{approximatedproblem} via bi-section search\\
        $\mathring{\Psi}_i^{t+1} \!\leftarrow \!\Psi_i^t \!+ \! \lambda_i \delta \left[ \mu_i D_{i}^t \!+\! \alpha_i {\rm e}^{\beta_i \Psi^{t}_i} \!+\! \gamma_i \! -\! \bar{\sigma}_i^t (\Psi_i^t \!-\!  \hat{\Psi}_i^t )\right] $\\
        \uIf{$\mathring{\Psi}_i^{t+1} > \bar{\Psi}$}{
            $D_i^t \leftarrow 0$\\ 
            $\Psi_i^{t+1} \leftarrow \Psi_i^t + \lambda_i \delta \left[  \alpha_i {\rm e}^{\beta_i \Psi^{t}_i} \!+\! \gamma_i - \bar{\sigma}_i^t (\Psi_i^t \!-\!  \hat{\Psi}_i^t ) \right]$\\
            $R_{\Psi_i} \leftarrow \bar{\Psi} - \mathring{\Psi}_i^{t+1}$
        }
        \uElseIf{$\Psi^{t, \text{\rm risk}}_i \le \mathring{\Psi}_i^{t+1} <  \bar{\Psi}$ }{
            $\Psi_i^{t+1} \leftarrow \mathring{\Psi}_i^{t+1}$\\
            $R_{\Psi_i} \leftarrow  \frac{D_i^t}{\text{std}(\{D_i^t,   i \in \mathcal{I} \})} + \Psi_{i}^{t, \text{risk}} - \mathring{\Psi}_i^{t+1}$
        } 
        \Else{
            $\Psi_i^{t+1} \leftarrow \mathring{\Psi}_i^{t+1}$\\
            $R_{\Psi_i} \leftarrow \frac{D_i^t}{\text{std}(\{D_i^t,   i \in \mathcal{I} \})}  $
        }
    }
    $R_t \leftarrow \sum_{i \in \mathcal{I}} R_{\Psi_i}$\\
    \Return {$\{D_i^t, {\Psi}_i^{t+1},   i \in \mathcal{I}\}$, $R_t$}\\
\end{algorithm} 


\section{Soft Actor-critic (SAC) Learning}

The state value function $Q(s_t, a_t)$ and the action-state value function $V(s_t)$ of the SAC are defined as follows:
\begin{align} 
    Q (s_t, a_t) &= R (s_t, a_t) + \gamma \mathbb{E}_{s_{t+1} \sim p(.|s_t, a_t)} \left[V (s_{t+1})\right], \label{eq:Q_value_function}\\
    V (s_t) &= \mathbb{E}_{a_t \sim \pi} \left[Q \left(s_t, a_t\right) - k \log\left(\pi \left(a_t | s_t\right)\right)\right]. \label{eq:V_function}
\end{align}

We employ function approximators for the policy function, the V-function, and the Q-function. Our approach includes a tractable policy $\pi_{\phi} (a_t | s_t)$, a parameterized state value function $V_{\psi} (s_t)$, and a soft Q-function $Q_{\theta} (s_t, a_t)$. Here, $\phi$, $\psi$, and $\theta$ denote network parameters. Our algorithm integrates three DNN types: V-network, Q-network, and policy network. Stochastic gradient descent is used to optimize the networks iteratively. 

To detail the parameter updates, we first focus on the update of the soft value that comes from the approximation of the state value function. The soft value function is trained by minimizing the squared residual error
\begin{align}
    J_{V}(\psi)=\mathbb{E}_{s_{t} \sim \mathcal{D}}\big [\frac{1}{2}\big ((V_{\psi}(s_{t}) -\mathbb{E}_{{a}_{t} \sim \pi_{\phi}}\big[Q_{\theta}(s_{t}, {a}_{t}) \notag\\
    -\log \pi_{\phi}({a}_{t} | s_{t})\big]\big )^{2}\big ],\label{eq:v}
\end{align}
where $\mathcal{D}$ is a replay buffer. Then, the gradient of the equation (\ref{eq:v}) is estimated using an unbiased estimator
\begin{align}
    &\hat{\nabla}_{\psi} J_{V}(\psi)=\nabla_{\psi} V_{\psi}\left(s_{i}\right)\big [V_{\psi}\left(s_{t}\right)\notag\\
    &\quad\quad \quad\quad\quad\quad\quad -Q_{\theta}\left(s_{t}, {a}_{t}\right)+\log \pi_{\phi}\left({a}_{t} | s_{t}\right)\big ],\label{eq:1}
\end{align}
wherein actions are chosen from the current policy set.

Furthermore, the soft Q-function parameter is trained by minimizing the following soft Bellman residual:
\begin{equation}
J_{Q}(\theta)=\mathbb{E}_{\left(s_{t}, {a}_{t}\right) \sim \mathcal{D}}\left[\frac{1}{2}\left(Q_{\theta}\left(s_{t}, {a}_{t}\right)-\hat{Q}\left(s_{t}, {a}_{t}\right)\right)^{2}\right],\label{eq:Q}
\end{equation}
with $\hat{Q}\left(s_{t}, {a}_{t}\right)=R\left(s_{t}, {a}_{t}\right)+\gamma \mathbb{E}_{s_{t+1} \sim p}\left[V_{\bar{\psi}}\left(s_{t+1}\right)\right]$. We use the following stochastic gradients to optimize \eqref{eq:Q}:
\begin{equation}
    \begin{aligned}
    \hat{\nabla}_{\theta} J_{Q}(\theta)=\nabla_{\theta} Q_{\theta}\left({a}_{t}, s_{t}\right)\text { ( }Q_{\theta}\left(s_{t}, {a}_{t}\right)\\
    -R\left(s_{t}, {a}_{t}\right)-\gamma V_{\bar{\psi}}\left(s_{t+1}\right)\text { ) },\label{eq:2}
    \end{aligned}
\end{equation}
where a target value network $V_{\bar{\psi}}$ is used for updates. We update the target weights to match the current value function weights periodically, i.e., 
\begin{equation}
    \bar{\psi} \leftarrow \tau \psi+(1-\tau) \bar{\psi},
\end{equation}
where $\tau$ is a target smoothing coefficient to improve stability. Finally, the policy parameter is learned by minimizing the expected Kullback-Leibler divergence:
\begin{equation}
    J_{\pi}(\phi)=\mathbb{E}_{s_{t} \sim \mathcal{D}}\left[\mathrm{D}_{\mathrm{KL}}\left(\pi_{\phi}\left(\cdot | s_{t}\right) \bigg \| \frac{\exp \left(Q_{\theta}\left(s_{t}, \cdot\right)\right)}{Z_{\theta}\left(s_{t}\right)}\right)\right].
\end{equation}
We use neural network transformation to reparameterize the policy, i.e.,
\begin{equation}
    {a}_{i}=g_{\phi}\left(\epsilon_{t} ; s_{t}\right),
\end{equation}
where $\epsilon_{t} $ is a noise vector. The objective can be rewritten as
\begin{align}
J_{\pi}(\phi)=\mathbb{E}_{s_{t} \sim \mathcal{D}, \epsilon_{t} \sim \mathcal{N}}\big[\log \pi_{\phi}\left(g_{\phi}\left(\epsilon_{t} ; s_{t}\right) | s_{t}\right)\notag\\
-Q_{\theta}\left(s_{t}, g_{\phi}\left(\epsilon_{t} ; s_{t}\right)\right)\big].\label{eq:PI}
\end{align}
The gradient of \eqref{eq:PI} can be approximated as
\begin{align}
    &\hat{\nabla}_{\phi} J_{\pi}(\phi)=\big (\nabla_{{a}_{t}} \log \pi_{\phi}\left({a}_{t} | s_{t}\big) \!-\!\nabla_{{a}_{t}} Q\left(s_{t}, {a}_{t}\right)\right) \nabla_{\phi} g_{\phi}\left(\epsilon_{t} ; s_{t}\right)\notag\\
    &\qquad\qquad\qquad\qquad\qquad\qquad\quad\quad+\nabla_{\phi} \log \pi_{\phi}\left({a}_{t} | s_{t}\right).\label{eq:3}
\end{align}

The training algorithm for SAC-based online thermal management for PCBS is summarized in Algorithm \ref{al:SAC}. Note that dual Q-networks are employed to mitigate positive bias during policy enhancement. Specifically, parameters $\theta_1$ and $\theta_2$ are used to shape two distinct Q-functions. They are trained individually, optimizing soft Bellman residual with parameters $J_Q (\theta_1)$ and $J_Q (\theta_2)$. Upon convergence, we obtain an online throughput management policy $\pi_{\phi}(a_t|s_t)$.

\begin{algorithm}[tbp]
\caption{Training for the SAC-based method}\label{al:SAC}
\KwIn{$\psi$, $\bar{\psi}$, $\phi$, $\theta_{1}$, $\theta_{2}$}
\KwOut{Trained policy $\pi_{\phi}(a_t|s_t)$}

Set learning rates $\lambda_{V}$, $\lambda_{Q}$, and $\lambda_{\pi}$ for V-network, Q-network, and policy network, respectively\\
Initialize $\psi$, $\bar{\psi}$, $\phi$, $\theta_{1}$, $\theta_{2}$ and experience memory $\mathcal{D}$\\

\For{\rm each episode}{
    \For{\rm each environment step}{
        ${a}_{t} \sim \pi_{\phi}\left({a}_{t} | s_{t}\right)$\\
        $s_{t+1} \sim p\left(s_{t+1}| s_{t}, {a}_{t}\right)$\\
        Obtain reward $R_t$ by Algorithm \ref{al:mechanism}\\
        $\mathcal{D} \leftarrow \mathcal{D} \cup\left\{\left(s_{t}, {a}_{t}, {R}_{t}, s_{t+1}\right)\right\}$
    }
    \For{\rm each gradient step}{
        Sample from current policy and compute $\hat{\nabla}_{\psi} J_{V}(\psi)$ by (\ref{eq:1})\\
        Update V-network parameter: $\psi \leftarrow \psi-\lambda_{V} \hat{\nabla}_{\psi} J_{V}(\psi)$ \\
        Sample from $\mathcal{D}$ and compute $\nabla J_{Q}(\theta_{n}), n \in \left\{1,2\right\}$ by (\ref{eq:2})\\
        Update dual Q-network parameters: $\theta_{n} \leftarrow \theta_{n}-\lambda_{Q} \hat{\nabla}_{\theta_{n}} J_{Q}\left(\theta_{n}\right),n \in \left\{1,2\right\}$\\
        Sample from the fixed distribution and compute $\hat{\nabla}_{\phi} J_{\pi}(\phi)$  by (\ref{eq:3})\\
        Update policy network parameter: $\phi \leftarrow \phi-\lambda_{\pi} \hat{\nabla}_{\phi} J_{\pi}(\phi)$ \\
        Update the target value network parameters: $\bar{\psi} \leftarrow \tau \psi+(1-\tau) \bar{\psi}$
    }
}

\Return{$\pi_{\phi}(a_t|s_t)$}
\end{algorithm}


\section{Simulation Results}

The hyperparameters of the proposed SAC scheme and the system model parameters are summarized in Tables \ref{tab:Margin_settings} and \ref{tab:System_model_settings}, respectively. We use the upper bound approximation introduced in Section~\ref{sec:noar} for performance evaluation because the approximation means we are on the safe side from thermal standpoint.

\begin{table}[ht]
    \caption{\label{tab:Margin_settings}SAC hyperparameters}
    \footnotesize
    \begin{center}
        \begin{tabular}{*{2}{l}}
            \toprule
            \midrule
            \bf{Hyperparameter}&\bf{Value}\\
            \hline
            Layers&	2 fully connected\\
            Layer hidden units	&  $64$\\
            Activation function  &  ReLU\\
            Batch size  &  256\\
            Replay buffer size  &  $1\times 10^6$\\
            Target smoothing coefficient  &   $0.005$\\
            Target update interval & $1$\\
            Discount rate  & $0.99$ \\
            Learning iterations per round  & $1$ \\
            Learning rate  & $3\times 10^{-4}$ \\
            Optimizer  & Adam  \\
            Loss  & Mean squared error  \\
            Entropy target factor & $0.2$ \\
            \midrule
            \bottomrule
        \end{tabular}
    \end{center}
\end{table}

\begin{table}[h]
    \caption{\label{tab:System_model_settings}System model parameters}
    \footnotesize
    \begin{center}
        \begin{tabular}{*{2}{l}}
            \toprule
            \midrule
            \bf{Parameter} & \bf{Value}\\
            \hline
            Number of PCBSs $I$ & $7$\\
            Number of users per cell $J_i$ & $100$\\
            Number of time slots $T$ & $100$\\
            Bandwidth of one RB $B$ & $180$kHz\\
            Safe temperature $\bar{\Psi}$ & $120 ^{\circ}$C \\
            Maximum throughput $D^{\text{max}}$ & $100$ Mbps\\
            Heat dissipation distribution & ${\sigma} \sim  U [0.25, 1.25] W/^{\circ}$C \\            
            Average ambient temperature $\tilde{\Psi}$& $[16,32] ^{\circ}$C\\
            Ambient temperature distribution & $\hat{\Psi} \sim U [0.8 \tilde{\Psi}, 1.2 \tilde{\Psi}] ^{\circ}$C \\
            Reciprocal of thermal capacitance $\lambda$ & $0.007 ^{\circ}$C/J \cite{BBUcooling}\\
            Throughput/dynamic-power ratio $\mu$ & $0.6$ W/Mbps \cite{dyn}\\
            Duration of a time slot $\delta$ & $30$ seconds\\
            \midrule
            \bottomrule
        \end{tabular}
    \end{center}
\end{table}

\subsection{Convergence}

We assess the convergence of our SAC-based thermal management (SAC-TM) strategy for both IHD and UHD scenarios, along with a comparison with the proximal policy optimization (PPO) algorithm \cite{schulman2017proximal}, referred to as PPO-TM. Note that the PPO algorithm uses the same configuration, including input design and reward functions, as SAC-TM. Figure~\ref{fig:reward} illustrates the average reward over episodes. The results indicate that SAC-TM tends to converge more rapidly than PPO-TM in both scenarios. Furthermore, it can be observed that in the IHD scenario, where the system has more known information, we obtain better performance, albeit a greater number of episodes for learning is required.

\input{FIG_result_reward_line_chart}

\subsection{Policy Performance Evaluation}

We assess the effectiveness of the SAC-TM policy by comparing it to what we refer to as Oracle. In Oracle, the global optimum is computed offline by convex optimization tools, such as CVX \cite{cvx}, by assuming prior knowledge of complete information across all time slots. Figure~\ref{fig:oracle} shows that the average throughput per cell with respect to the (average) ambient temperature. The average is taken for 1000 instances for each value of average ambient temperature.  The results clearly demonstrate that SAC-TM achieves very good performance, and outperforms PPO-TM. For an average ambient temperature of $16^{\circ}$C, the policy yields $86.9\%$ and $82.6\%$ of the global optimum for the IHD and UHD scenarios, respectively. Higher ambient temperature results in lower throughput as expected, but the (absolute) throughput difference to the global optimum remains largely the same.

\pgfplotsset{width=.49\textwidth, height=5.5cm, compat=1.9}
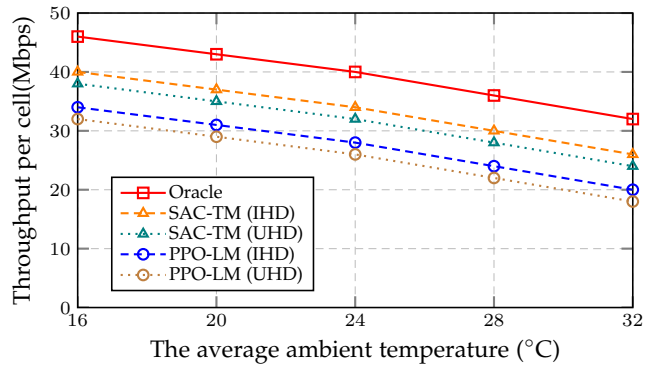
\begin{figure}[t]
        \begin{center}
        \scriptsize
    		\begin{tikzpicture}
        		\begin{axis}[
        		    xlabel={The average ambient temperature ($^\circ$C)},
        		    ylabel={Throughput per cell(Mbps)},
        		    xmin=16, xmax=32,
        		    ymin=0, ymax=50,
        		      xtick={16, 20, 24, 28, 32 },
        		    ytick={0, 10, 20, 30, 40, 50},
                        legend style={at={(0.06, 0.25)},anchor=west},
        		    ymajorgrids=true,    
        		    xmajorgrids=true,
        		    grid style=densely dashed,
        		    tick label style={font=\scriptsize},
        		    label style={font=\small},
        		    legend style={font=\scriptsize},
        		]
        		
                    \addplot[ color=red, mark=square, line width=0.8pt]     
            		coordinates { 
                    ( 16 , 46)
                    ( 20 , 43 )
                    ( 24 , 40 )
                    ( 28 , 36 )
                    ( 32 , 32 )
            		};

                    \addplot[ color=orange, mark=triangle, densely dashed, mark options={solid}, line width=0.8pt]
            		coordinates {
                    ( 16 , 40)
                    ( 20 , 37 )
                    ( 24 , 34 )
                    ( 28 , 30 )
                    ( 32 , 26 )
            		}; 
            		
                    \addplot[ color=teal, mark=triangle, dotted, mark options={solid}, line width=0.8pt]
            		coordinates {
                    ( 16 , 38)
                    ( 20 , 35 )
                    ( 24 , 32 )
                    ( 28 , 28 )
                    ( 32 , 24 )
            		};            		
                                  \addplot[ color=blue, mark=o, densely dashed, mark options={solid}, line width=0.8pt]
            		coordinates {
                    ( 16 , 34)
                    ( 20 , 31 )
                    ( 24 , 28 )
                    ( 28 , 24 )
                    ( 32 , 20 )
            		};            		
                    
                    \addplot[ color=brown, mark=o, dotted, mark options={solid}, line width=0.8pt]
            		coordinates {
                    ( 16 , 32)
                    ( 20 , 29 )
                    ( 24 , 26 )
                    ( 28 , 22 )
                    ( 32 , 18 )
            		};

            		\legend{Oracle, SAC-TM (IHD), SAC-TM (UHD), PPO-LM (IHD), PPO-LM (UHD)}
        		
        		\end{axis}
    		\end{tikzpicture}
    		
        \end{center}
    \caption{Throughput per cell with respect to the average ambient temperature. (The average heat dissipation efficiency is $0.75$.)}\label{fig:oracle}
\end{figure}

For the IHD scenario, Figure \ref{fig:temperature_result} illustrates the maximum BBU chip temperature (of all seven PCBSs) and the total throughput under three policies: SAC-TM, PPO-TM, and a naive adaptive policy (cf.~Section~\ref{subsec:PCBS}) for one data instance. For each time slot, this policy considers the maximum cell-uniform throughput that is permitted by the resource constraint. This throughput level is used if the resulting temperature by the end of the time slot is within the limit for all cells; otherwise, the policy takes no action to let the BBU cool down. Overall, SAC-TM clearly outperforms the other two. With SAC-TM, the temperature stays close but below the limit.  With the naive adaptive policy, the temperature clearly has more fluctuations, and the throughput is considerably lower than that via SAC-TM. PPO-TM performance is somewhere between SAC-TM and the native adaptive policy.

\input{FIG_result_temparature_throughput_line_chart}

\subsection{Reward Mechanism Assessment}

A basic reward function for our problem is simply letting the reward equal the total throughput, i.e., $R(s_t, a_t) = \sum_{i \in \mathcal{I}} \sum_{t \in \mathcal{T}} D_{i}^t$. In Fig.~\ref{fig:Learning}, we compare this basic reward function to the one proposed in Section~\ref{sec:temp}, by showing the agent's longevity in terms of steps taken in the UHD scenario in the training session. For each training session, the step count is 100. The training concludes prematurely if the agent's actions breach the temperature constraint \eqref{T_C2} or the resource constraint \eqref{LC_C2}.  Our designed reward function enables the agent to consistently achieve an average of 100 steps of longevity, whereas SAC with the basic reward function, referred to as the naive SAC; fails to reach this number even after 100,000 episodes. This underscores the effectiveness of the proposed SAC in guiding the agent. Note that the performance curve of the proposed SAC shows slight fluctuation between the 35,000th and 40,000th episodes. This is due to the UHD scenario, where the agent has to rely on the estimation of heat dissipation efficiency.  There exists a possibility that the actual temperature development may exceed the estimated one, potentially leading to overheating. Table~\ref{tab:results} presents the overheating rate and throughput for the two policies for $1000$ instances. The results demonstrate that the proposed SAC is superior.

\input{FIG_result_step_line_chart}

\begin{table}[ht]
    \caption{\label{tab:results} Results for verifying the proposed mechanism}
    \footnotesize
    \begin{center}
        \begin{threeparttable}[b]
            \begin{tabular}{*{3}{c}}
                \toprule
                \midrule
                & {\bf Proposed SAC} &{\bf Naive SAC} \\
                \midrule
                Overheating Rate & $0.1\%$& $12.5\%$    \\
                \midrule
                Throughput (Mbps) & $33.2$ & $24.1$  \\
                \midrule
                \bottomrule
            \end{tabular}
        \end{threeparttable}
    \end{center}
\end{table}

\subsection{Generalization for  Mobility}

The proposed RL is intended for online use, and hence it will continuously learn and adapt the policy to user mobility. Nevertheless, it is of significance to examine how well a learned policy generalizes with respect to mobility. To evaluate this aspect, we compare the performance of a policy for a new scenario where some users' positions are changed, versus that of training a specific model tailored to the new scenario, for various degrees of user mobility. The degree of mobility is quantified in the percentage of cell diameter. That is, all users move in some random direction, with a distance specified in percentage of cell diameter.  The results are presented in Fig.~\ref{fig:generalication} for a representative instance. For up to 15\% of the degree of mobility, the initial policy works well in throughput as well as the risk of overheating. Thus, the proposed scheme is sufficiently robust for online learning and decision-making.

\pgfplotsset{width=.45\textwidth, height=5.5cm, compat=1.9}
\begin{figure}[t]
    \begin{tikzpicture}
        \centering
        \begin{axis}[
                xtick style={draw=none}, ytick pos=left, 
                ybar, axis on top,
                tick label style={font=\scriptsize},
                xtick align=inside,
                xlabel={Degree of mobility},
                ylabel style={align=center},
                ylabel={Throughput per cell (Mbps)},
                ymin = 0, ymax = 48,
                symbolic x coords={1, 2, 3, 4, 5},
                xticklabels={, $5\%$, $10\%$, $15\%$, $20\%$, $25\%$},
                enlarge x limits = 0.15,
                legend style = {at = {(0.03, 0.97)},
                anchor = north west, 
                legend columns = -1},
                legend image code/.code={\draw [#1] (0cm,-0.1cm) rectangle (0.2cm,0.15cm); },
                bar width=.016\textwidth,
                tick label style={font=\scriptsize},
                label style={font=\small},
                legend style={font=\scriptsize},
                minor x tick num=1,
                xminorgrids,
                minor tick length=0,
                grid style={dashed},
        ]
            \addplot  coordinates { (1, 32) (2, 32) (3, 32) (4, 27) (5, 15)};
            \addplot  coordinates { (1, 32.3) (2, 32.5) (3, 33.1) (4, 33.3) (5, 33.1)};
            \legend{Initial, Specific}
        \end{axis}
		\begin{axis}[
    		xtick style={draw=none}, ytick pos=right, 
		    ylabel={Non-overheating rate (\%)},
		    xticklabels=\empty,
		    ymin = 0, ymax=130,
		    ytick = {0, 50, 100},
                enlarge x limits = 0.15,
                legend style = {at = {(0.5, 0.97)},
                anchor = north west, 
                legend columns = -1},
		    grid style=densely dashed,
		    tick label style={font=\scriptsize},
		    label style={font=\small},
		    legend style={font=\scriptsize},
		]
                \addplot[ color=teal, mark=triangle, line width=0.8pt]  
    		coordinates {(1, 99.5) (2, 99) (3, 98) (4, 88) (5, 61) };
                \addplot[ color=orange, mark=o, line width=0.8pt]     
    		coordinates {(1, 100) (2, 100) (3, 100) (4, 100) (5, 100) };
                \legend{Initial, Specific}
		\end{axis}
    \end{tikzpicture}
    \caption{Performance comparison of initial and new policies
for various degrees of user position update.} \label{fig:generalication}
\end{figure}
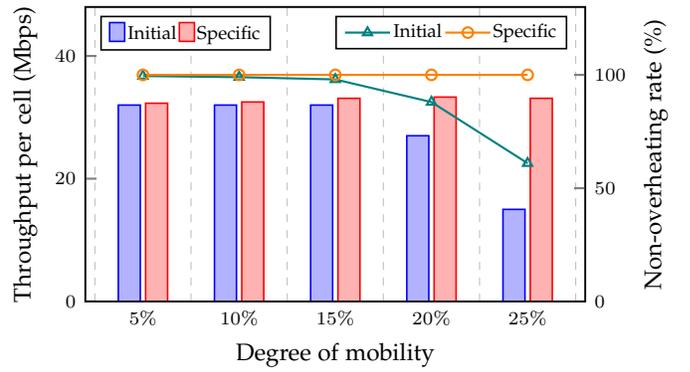


\section{Conclusions}

For outdoor PCBSs, 
the trade-off between utilization (i.e., throughput) and
thermal development is vital. The online optimization problem we
studied presents two challenges: Uncertain (future) heat dissipation
and interference coupling. As our main conclusion, provided that the
reward mechanism is specifically tailored to the problem, an RL
approach based on SAC exhibits strong performance in delivering high
throughput while keeping the BBU temperature beneath but close to the
maximum limit.  In fact, its performance is close to that of offline
global optimum, hence there is not much room for
improvement. Moreover, the approach generalizes well with respect to
mobility, enabling it to be embedded into an online learning and
decision-making framework. For future study, one topic of interest is
to incorporate additional mechanisms beyond RL itself, such as the
integration of RL and robust optimization techniques.


\appendices

\section{Proof of Lemma~\ref{lemma:sif}} \label{sec:proofsif}
\begin{proof}
First, it is obvious that function $f_{ij}(\boldsymbol{\rho}_{-ij})$ has monotonicity, that is, $f_{ij}(\boldsymbol{\rho}'_{-ij}) \geq f_{ij}(\boldsymbol{\rho}_{-ij})$, if $\boldsymbol{\rho}'_{-ij} \geq \boldsymbol{\rho}_{-ij}$. Next, we show that the function also exhibits the so called scalability property, which requires $\alpha f_{ij}(\boldsymbol{\rho}_{-ij}) > f_{ij}(\alpha \boldsymbol{\rho}_{-ij})$ for any scalar $\alpha > 1$. 

Let $\rho^\Xi = \sum_{\ell \in \mathcal{I} \setminus \{i\}} \sum_{k \in \mathcal{J}_\ell} \rho_{\ell k} h_{\ell k j}$. Then, by the expression of $f_{ij}(\boldsymbol{\rho}_{-ij})$, scalability holds, if 

\begin{align}
    \alpha \log \left[ 1+ \frac{h_{ij}}{\displaystyle \alpha \rho^\Xi + N_c} \right] > 
\log \left[ 1+ \frac{h_{ij}}{\displaystyle \rho^\Xi + N_c} \right].
\end{align}

The above holds if 

\begin{align}
\label{eq:scala}
    \left[ 1+ \frac{h_{ij}}{\displaystyle \alpha \rho^\Xi + N_c} \right]^\alpha > 
 1+ \frac{h_{ij}}{\displaystyle \rho^\Xi + N_c} .
\end{align}

By the generalized Bernoulli inequality (with real exponent), we have 

\begin{align}
\label{eq:bern}
    \left[ 1+ \frac{h_{ij}}{\displaystyle \alpha \rho^\Xi + N_c} \right]^\alpha > 
 1 + \frac{\alpha h_{ij}}{\displaystyle \alpha \rho^\Xi + N_c} = 
1 + \frac{h_{ij}}{\displaystyle \rho^\Xi + \frac{N_c}{\alpha}}.
\end{align}

By comparing the right-hand sides of \eqref{eq:scala} and \eqref{eq:bern}, we have proven scalability. Hence $f_{ij}(\boldsymbol{\rho}_{-ij})$ is an SIF~\cite{yates1995framework}.
\end{proof}

\section{Proof of SIF with the Approximations}
\label{sec:proofasif}
\begin{proof}
With either of the two approximations, we obtain the system of the following form.
\begin{align}\label{eq:cellrhofunction}
f_{ij}(\boldsymbol{\rho}_{-i}) = \frac{D_i}{B |\mathcal{J}_i| \log \left[ 1+ \frac{h_{ij}}{ \sum\limits_{\ell \in \mathcal{I} \setminus \{i\}} \rho_{\ell} \gamma_\ell + BN_0} \right]}
\end{align}

For the long range approximation, we have $\gamma_\ell = E|g_{\ell j}|^2$, and in the upper bound approximation, $\gamma_\ell = \frac{E}{N_T}\|\mathbf{H}_{lj}\mathbf{x}_{lj}^\bigstar \|^2$.  It is clear that \eqref{eq:cellrhofunction} has a similar (but simpler) structure to \eqref{eq:rhofunction}. The SIF result is then obtained by the same arguments used in the proof of Lemma~\ref{lemma:sif}.
\end{proof}


\bibliographystyle{IEEEtran}
\bibliography{mybibtex}

\begin{thebibliography}{10}
\providecommand{\url}[1]{#1}
\csname url@samestyle\endcsname
\providecommand{\newblock}{\relax}
\providecommand{\bibinfo}[2]{#2}
\providecommand{\BIBentrySTDinterwordspacing}{\spaceskip=0pt\relax}
\providecommand{\BIBentryALTinterwordstretchfactor}{4}
\providecommand{\BIBentryALTinterwordspacing}{\spaceskip=\fontdimen2\font plus
\BIBentryALTinterwordstretchfactor\fontdimen3\font minus
  \fontdimen4\font\relax}
\providecommand{\BIBforeignlanguage}[2]{{%
\expandafter\ifx\csname l@#1\endcsname\relax
\typeout{** WARNING: IEEEtran.bst: No hyphenation pattern has been}%
\typeout{** loaded for the language `#1'. Using the pattern for}%
\typeout{** the default language instead.}%
\else
\language=\csname l@#1\endcsname
\fi
#2}}
\providecommand{\BIBdecl}{\relax}
\BIBdecl

\bibitem{alliance2021network}
{Next Generation Mobile Networks}, ``Green network future: Network energy
  efficiency,'' 2021,
  \url{https://www.ngmn.org/wp-content/uploads/211009-GFN-Network-Energy-Efficiency-1.0.pdf}.

\bibitem{ericsson22}
\BIBentryALTinterwordspacing
Ericsson, ``Ericsson {5G} portfolio update puts energy efficiency center
  stage,'' 2022. [Online]. Available:
  \url{https://mb.cision.com/Main/15448/3507058/1536906.pdf}
\BIBentrySTDinterwordspacing

\bibitem{haarnoja2018soft}
T.~Haarnoja, A.~Zhou, P.~Abbeel, and S.~Levine, ``Soft actor-critic: Off-policy
  maximum entropy deep reinforcement learning with a stochastic actor,'' in
  \emph{International conference on machine learning}.\hskip 1em plus 0.5em
  minus 0.4em\relax PMLR, 2018, pp. 1861--1870.

\bibitem{liu2023passive}
H.~Liu, J.~Yu, C.~Wang, Z.~Zeng, P.~Poredo{\v{s}}, and R.~Wang, ``Passive
  thermal management of electronic devices using sorption-based evaporative
  cooling,'' \emph{Device}, vol.~1, no.~6, 2023.

\bibitem{sui2023membrane}
Z.~Sui, Y.~Sui, Z.~Ding, H.~Lin, F.~Li, R.~Yang, and W.~Wu,
  ``Membrane-encapsulated, moisture-desorptive passive cooling for
  high-performance, ultra-low-cost, and long-duration electronics thermal
  management,'' \emph{Device}, vol.~1, no.~6, 2023.

\bibitem{Aslan2019Passive}
Y.~Aslan, C.~E. Kiper, A.~Johannes van~den Biggelaar, U.~Johannsen, and
  A.~Yarovoy, ``Passive cooling of mm-wave active integrated {5G} base station
  antennas using {CPU} heatsinks,'' in \emph{2019 16th European Radar
  Conference (EuRAD)}, 2019, pp. 121--124.

\bibitem{Aslan2019Heat}
Y.~Aslan, J.~Puskely, A.~Roederer, and A.~Yarovoy, ``Heat transfer enhancement
  in passively cooled {5G} base station antennas using thick ground planes,''
  in \emph{2019 13th European Conference on Antennas and Propagation (EuCAP)},
  2019, pp. 1--5.

\bibitem{Duan2021Thermal}
K.-W. Duan, J.-W. Shi, and W.-Q. Tao, ``Thermal design for the passive cooling
  system of radio base station with high power density,'' in \emph{Advances in
  Heat Transfer and Thermal Engineering: Proceedings of 16th UK Heat Transfer
  Conference (UKHTC2019)}, 2021, pp. 617--620.

\bibitem{Aslan1}
Y.~Aslan, J.~Puskely, A.~Roederer, and A.~Yarovoy, ``Multiple beam synthesis of
  passively cooled {5G} planar arrays using convex optimization,'' \emph{IEEE
  Transactions on Antennas and Propagation}, vol.~68, no.~5, pp. 3557--3566,
  2020.

\bibitem{Aslan2}
------, ``Effect of element number reduction on inter-user interference and
  chip temperatures in passively-cooled integrated antenna arrays for {5G},''
  in \emph{2020 14th European Conference on Antennas and Propagation (EuCAP)},
  2020, pp. 1--5.

\bibitem{data-center-1}
J.~Wan, X.~Gui, R.~Zhang, and L.~Fu, ``Joint cooling and server control in data
  centers: A cross-layer framework for holistic energy minimization,''
  \emph{IEEE Systems Journal}, vol.~12, no.~3, pp. 2461--2472, 2018.

\bibitem{c-ran-1}
S.~K. Sah~Tyagi, T.~Lin, and Y.~Zhou, ``Thermal-aware dynamic computing
  resource allocation for {BBU} pool in centralized radio access networks,'' in
  \emph{IEEE 85th VTC Spring}, 2017, pp. 1--5.

\bibitem{c-ran-2}
S.~K. Sah~Tyagi, Y.~Zhou, T.~Lin, A.~Marahatta, and J.~Shi, ``Realization of a
  computational efficient {BBU} cluster for cloud ran,'' \emph{Journal of
  Ambient Intelligence and Humanized Computing}, pp. 1--11, August 2018.

\bibitem{Ran2023}
Y.~Ran, H.~Hu, Y.~Wen, and X.~Zhou, ``Optimizing energy efficiency for data
  center via parameterized deep reinforcement learning,'' \emph{IEEE
  Transactions on Services Computing}, vol.~16, no.~2, pp. 1310--1323, 2023.

\bibitem{Bao2023Thermal}
L.~Bao, Z.~He, J.~Tan, Y.~Chen, and M.~Zhao, ``Thermal-aware task scheduling
  and resource allocation for uav-and-basestation hybrid-enabled mec
  networks,'' \emph{IEEE Transactions on Green Communications and Networking},
  vol.~7, no.~2, pp. 579--593, 2023.

\bibitem{alsuhli2021mobility}
G.~Alsuhli, K.~Banawan, K.~Attiah, A.~Elezabi, K.~Seddik, A.~Gaber, M.~Zaki,
  and Y.~Gadallah, ``Mobility load management in cellular networks: {A} deep
  reinforcement learning approach,'' \emph{IEEE Transactions on Mobile
  Computing}, 2021.

\bibitem{9855432}
A.~Feriani, D.~Wu, Y.~T. Xu, J.~Li, S.~Jang, E.~Hossain, X.~Liu, and G.~Dudek,
  ``Multiobjective load balancing for multiband downlink cellular networks: A
  meta-reinforcement learning approach,'' \emph{IEEE JSAC}, vol.~40, no.~9, pp.
  2614--2629, 2022.

\bibitem{Chang2020}
K.-C. Chang, K.-C. Chu, H.-C. Wang, Y.-C. Lin, and J.-S. Pan, ``Energy saving
  technology of {5G} base station based on internet of things collaborative
  control,'' \emph{IEEE Access}, vol.~8, pp. 32\,935--32\,946, 2020.

\bibitem{Israr2023}
A.~Israr, Q.~Yang, and A.~Israr, ``Renewable energy provision and
  energy-efficient operational management for sustainable 5g infrastructures,''
  \emph{IEEE Transactions on Network and Service Management}, vol.~20, no.~3,
  pp. 2698--2710, 2023.

\bibitem{Wu2021}
Q.~Wu, X.~Chen, Z.~Zhou, L.~Chen, and J.~Zhang, ``Deep reinforcement learning
  with spatio-temporal traffic forecasting for data-driven base station sleep
  control,'' \emph{IEEE/ACM Transactions on Networking}, vol.~29, no.~2, pp.
  935--948, 2021.

\bibitem{Amine2022}
A.~E. Amine, J.-P. Chaiban, H.~A.~H. Hassan, P.~Dini, L.~Nuaymi, and R.~Achkar,
  ``Energy optimization with multi-sleeping control in 5g heterogeneous
  networks using reinforcement learning,'' \emph{IEEE Transactions on Network
  and Service Management}, vol.~19, no.~4, pp. 4310--4322, 2022.

\bibitem{YuZhYoYu24}
Z.~Yu, Z.~Yi, Y.~Deng, and D.~Yuan, ``Learn to stay cool: online load
  management for passively cooled base stations,'' in \emph{2012 IEEE Wireless
  Communications and Networking Conference}.\hskip 1em plus 0.5em minus
  0.4em\relax IEEE, 2024.

\bibitem{majewski2010conservative}
K.~Majewski and M.~Koonert, ``Conservative cell load approximation for radio
  networks with shannon channels and its application to lte network planning,''
  in \emph{2010 Sixth Advanced International Conference on
  Telecommunications}.\hskip 1em plus 0.5em minus 0.4em\relax IEEE, 2010, pp.
  219--225.

\bibitem{fehske2012aggregation}
A.~J. Fehske and G.~P. Fettweis, ``Aggregation of variables in load models for
  interference-coupled cellular data networks,'' in \emph{2012 IEEE
  International Conference on Communications (ICC)}.\hskip 1em plus 0.5em minus
  0.4em\relax IEEE, 2012, pp. 5102--5107.

\bibitem{Siomina2012Analysis}
I.~Siomina and D.~Yuan, ``Analysis of cell load coupling for lte network
  planning and optimization,'' \emph{IEEE Transactions on Wireless
  Communications}, vol.~11, no.~6, pp. 2287--2297, 2012.

\bibitem{Awan2020Robust}
D.~A. Awan, R.~L.~G. Cavalcante, and S.~Stanczak, ``Robust cell-load learning
  with a small sample set,'' \emph{IEEE Transactions on Signal Processing},
  vol.~68, pp. 270--283, 2020.

\bibitem{klessig2015performance}
H.~Klessig, D.~{\"O}hmann, A.~J. Fehske, and G.~P. Fettweis, ``A performance
  evaluation framework for interference-coupled cellular data networks,''
  \emph{IEEE Transactions on Wireless Communications}, vol.~15, no.~2, pp.
  938--950, 2015.

\bibitem{fehske2013concurrent}
A.~J. Fehske, H.~Klessig, J.~Voigt, and G.~P. Fettweis, ``Concurrent load-aware
  adjustment of user association and antenna tilts in self-organizing radio
  networks,'' \emph{IEEE Transactions on Vehicular Technology}, vol.~62, no.~5,
  pp. 1974--1988, 2013.

\bibitem{cavalcante2016low}
R.~L.~G. Cavalcante, S.~Sta{\'n}czak, J.~Zhang, and H.~Zhuang, ``Low complexity
  iterative algorithms for power estimation in ultra-dense load coupled
  networks,'' \emph{IEEE Transactions on Signal Processing}, vol.~64, no.~22,
  pp. 6058--6070, 2016.

\bibitem{arani2023haps}
A.~H. Arani, P.~Hu, and Y.~Zhu, ``Haps-uav-enabled heterogeneous networks: A
  deep reinforcement learning approach,'' \emph{IEEE Open Journal of the
  Communications Society}, 2023.

\bibitem{BBUcooling}
T.~Heath, A.~P. Centeno, P.~George, L.~Ramos, Y.~Jaluria, and R.~Bianchini,
  ``Mercury and freon: Temperature emulation and management for server
  systems,'' \emph{ACM SIGOPS Operating Systems Review}, vol.~40, no.~5, pp.
  106--116, 2006.

\bibitem{dynsta}
B.~Goel, S.~A. McKee, and M.~Sj{\"a}lander, ``{Techniques} to measure, model,
  and manage power,'' ser. Advances in Computers.\hskip 1em plus 0.5em minus
  0.4em\relax Elsevier, 2012, vol.~87, ch.~2, pp. 7--54.

\bibitem{dyn}
F.~Marzouk, T.~Akhtar, I.~Politis, J.~P. Barraca, and A.~Radwan, ``Power
  minimizing {BBU-RRH} group based mapping in {C-RAN} with constrained
  devices,'' in \emph{IEEE ICC}, 2020, pp. 1--7.

\bibitem{sta}
K.~DeVogeleer, G.~Memmi, P.~Jouvelot, and F.~Coelho, ``Modeling the temperature
  bias of power consumption for nanometer-scale {CPUs} in application
  processors,'' in \emph{2014 International Conference on Embedded Computer
  Systems: Architectures, Modeling, and Simulation (SAMOS XIV)}, 2014, pp.
  172--180.

\bibitem{knopp1995information}
R.~Knopp and P.~A. Humblet, ``Information capacity and power control in
  single-cell multiuser communications,'' in \emph{Proceedings IEEE
  International Conference on Communications ICC'95}, vol.~1.\hskip 1em plus
  0.5em minus 0.4em\relax IEEE, 1995, pp. 331--335.

\bibitem{yates1995framework}
R.~D. Yates, ``A framework for uplink power control in cellular radio
  systems,'' \emph{IEEE Journal on selected areas in communications}, vol.~13,
  no.~7, pp. 1341--1347, 1995.

\bibitem{schulman2017proximal}
J.~Schulman, F.~Wolski, P.~Dhariwal, A.~Radford, and O.~Klimov, ``Proximal
  policy optimization algorithms,'' \emph{arXiv preprint arXiv:1707.06347},
  2017.

\bibitem{cvx}
{CVX Research Inc.}, ``{CVX}: Matlab software for disciplined convex
  programmicvxr.com/cvx.''

\end{thebibliography}

\end{document}